\documentclass[5p, twocolumn]{elsarticle}
\usepackage{natbib, hyperref, todonotes}
\usepackage[fleqn, tbtags]{amsmath}
\usepackage{amssymb}
\usepackage{mathtools}
\usepackage{lmodern, mathrsfs, BOONDOX-cal}
\usepackage[scr=cm, scrscaled=1.05]{mathalfa}
\usepackage{amsthm}
\usepackage[ruled, linesnumbered]{algorithm2e}

\SetCommentSty{commentFont}
\usepackage{graphicx}
\usepackage[outdir=./.temp/]{epstopdf}
\usepackage[justification=centering]{caption}
\usepackage{threeparttable, diagbox, multirow, makecell}
\usepackage{ulem, cases, enumitem, stackengine}
\usepackage{booktabs,xcolor,siunitx}
\usepackage{wrapfig}

\newtheorem{assumption}{Assumption}

\newtheorem{problem}{Problem}
\newtheorem{lemma}{Lemma}
\newtheorem{theorem}{Theorem}
\newtheorem{corollary}{Corollary}
\newtheorem{remark}{Remark}

\def\+{{\!+\!}}
\def\-{{\!-\!}}
\def\op[#1]{{\operatorname{#1}}}

\def\vec[#1]{{\boldsymbol{#1}}}
\def\mat[#1]{{\mathbf{#1}}}
\def\vecspace[#1]{{\mathscr{#1}}}
\def\tr{^{\mathrm{T}}} \def\H{^{\mathrm{H}}}
\def\E{\mathbb{E}}
\def\cov{{\operatorname{cov}}}
\def\N{\mathbb{N}} \def\Z{\mathbb{Z}}
\def\R{\mathbb{R}} \def\C{\mathbb{C}}
\def\im{{\mathcal{j}}}
\def\indk{{\mathcal{k}}} \def\indl{{\mathcal{l}}}
\def\indm{{\mathcal{m}}} \def\indn{{\mathcal{n}}}
\def\indp{{\mathcal{p}}} \def\indq{{\mathcal{q}}}
\def\indr{{\mathcal{r}}} \def\inds{{\mathcal{s}}}

 \setlist{noitemsep}   \allowdisplaybreaks[2]  \newlength{\casealign}  \newcommand{\printLongest}[1]{\settowidth{\global\casealign}{$#1$}#1}
\newcommand{\continueToAlign}[1]{\mathrlap{#1}\hspace{\casealign}}

\begin{document}

\title{Combined Invariant Subspace \& Frequency-Domain Subspace Method for Identification of Discrete-Time MIMO Linear Systems \tnoteref{fund}}
\author[1]{Jingze You\fnref{fn1}}\ead{2133053@tongji.edu.cn}
\author[1]{Chao Huang\fnref{fn1}}\ead{csehuangchao@tongji.edu.cn}
\author[1,2]{Hao Zhang\fnref{fn2}}\ead{zhang_hao@tongji.edu.cn}
\affiliation[1]{organization={College of Electronic and Information Engineering, Tongji University}, postcode={200092},city={Shanghai}, country={China}}
\affiliation[2]{organization={National Key Laboratory of Autonomous Intelligent Unmanned Systems \& Frontiers Science Center for Intelligent Autonomous Systems, Tongji University}, p={200092}, c={Shanghai}, cy={China}}

\fntext[fn1]{These authors contributed equally to this work and should be considered co-first authors.}
\fntext[fn2]{Corresponding author.}
\fntext[fn3]{The final version of this paper has been accepetd by \emph{Systems \& Control Letters} with DOI: \href{https://doi.org/10.1016/j.sysconle.2023.105641}{10.1016/j.sysconle.2023.105641}.}
\tnotetext[fund]{This work was supported in part by the National Natural Science Foundation of China (Grant No. 61922063, 62273255, 62150026), and in part by the Fundamental Research Funds for the Central Universities.}

\begin{abstract}
    Recently, a novel system identification method based on invariant subspace theory is introduced, aiming to address the identification problem of continuous-time (CT) linear time-invariant (LTI) systems by combining time-domain and frequency-domain methods.
    Subsequently, the combined Invariant-Subspace and Subspace Identification Method (cISSIM) is introduced, enabling direct estimation of CT LTI systems in state-space forms.
    It produces consistent estimation that is robust in an error-in-variable and slow-sampling conditions, while no pre-filtering operation of the input-output signals is needed.
    This paper presents the discrete-cISSIM, which extends cISSIM to discrete-time (DT) systems and offers the following improvements:
    1) the capability to utilize arbitrary discrete periodic excitations while cISSIM uses multi-sine signals;
    2) a faster estimation with reduced computational complexity is proposed;
    3) the covariance estimation problem can be addressed concurrently with the system parameter estimation.
    An implementation of discrete-cISSIM by MATLAB has also been provided.
\end{abstract}
\maketitle

\section{Introduction}
System identification of black-box linear time invariant (LTI) plants is an established research area.
Most of the existing methods belong to two categories: \textit{time-domain method} and \textit{frequency-domain method} \cite{SystemIdentification_Ljung_1}.

Time-domain methods directly utilize the sampled data to identify the system.
One type of time-domain methods is based on least-squares, including Box-Jenkins method for SISO systems \cite{SystemIdentification_Iserman_1, NumericalLeastSquare_Bjoorck_1}, and subspace identification method (SIM) for MIMO systems \cite{SID_Overschee_1}. 
Another thoughts include methods based on the statistics, such as expectation maximization methods \cite{CovarianceEstimationML_Shumway_1}.
The advantage of time-domain methods is that they can be applied to both continuous-time (CT) and discrete-time (DT) systems.
However, most of them require additional procedures to resolve the error-in-variable (EIV) problem, such as instrumental variables, bias-eliminating least-squares and total least-squares \cite{EIVProblem_Soderstrom_1}.

Frequency-domain method uses frequency data that are typically obtained through discrete Fourier transformation (DFT) or auto-correlation function applied to the time-domain sampled data.
The parameter estimation usually ends up with a nonlinear optimization problem, such as frequency SIM \cite{SysIden_FreqSubspace_McKelvey_1}.
Frequency-domain methods offer an easy solution to the error-in-variable problem.
However, it should be careful to use them in real-time because of the spectral leakage and aliasing problems from short-time Fourier transformation (STFT) \cite{FreqSystemIdentification_Pintelon_1, SysIden_FreqTransient_pintelon_1}.

The basis of this paper is the Invariant-SubsPace (ISP) method which combines the time-domain and frequency-domain methods.
There have been two published works based on ISP: the \textbf{basic ISP method} \cite{cISSIM_CHuang_1} and the \textbf{combined Invariant-Subspace and Subspace Identification Method (cISSIM)} \cite{cISSIM_CHuang_2}.
The former uses transfer function model and least-squares method to estimate.
The latter uses state-space to model the MIMO system, and the frequency SIM method to estimate the system parameter.
The advantages of cISSIM compared to antecedent methods includes:
1) intrinsic robustness against slow sampling rate and the EIV problem;
2) time-domain samples used directly without pre-filtering.
The method in this paper called \textbf{discrete-cISSIM} extends the cISSIM into DT domain and arbitrary discrete periodic excitations with reduced computational burden.

In addition to system parameters, many applications also require the statistical information of noises \cite{StateEstimation_Simon_1}, such as state estimation.
The existing statistical estimation methods can be classified according to the noise model:
The first type estimates the covariance matrices, but needs the system matrices as \textit{a priori}, such as autocovariance least square (ALS) \cite{ALSMatrix_Odelson_1, ALSVector_Dunik_1}.
The second type can only obtain the innovation model, but can estimate the system and noise parameters together, such as Bayesian \cite{CovarianceEstimationBayesian_Lainiotis_1}, maximum likelihood \cite{CovarianceEstimationML_Shumway_1}, and subspace methods \cite{ARMAModel_Faurre_1, ALSProblem_Kong_2, CovarianceEstimation_Michel_1}.
Up to authors' best knowledge, it is still \textbf{an open problem} for the black-boxed combined estimation of both system parameters and covariance matrix estimation.
Furthermore, the input noise has been not taken into consideration.
This paper combines the ALS method with the existing cISSIM method to solve the above problems, and makes the calculations in frequency domain to speed it up.

Based on combined deterministic-stochastic subspace identification \cite{SID_Overschee_2}, this paper derives an algorithm that is able to estimate the system and covariance parameters altogether.
The \textbf{highlights} of this paper are concluded as follows:
\begin{itemize}
    \item[$\bullet$] (inherited) identify the system with EIV issues;
    \item[$\bullet$] ability to use arbitrary discrete periodic signals;
    \item[$\bullet$] reduced computational burden without precision loss;
    \item[$\bullet$] identify the covariance matrix together, and speed it up by calculating it in frequency-domain.
\end{itemize}

Additionally, an implementation of the discrete-cISSIM including simulation tests via MATLAB has been uploaded to GitHub: \href{https://github.com/wyqy/dcissim}{https://github.com/wyqy/dcissim}.

The remaining part includes the following:
Section \ref{section_problem} gives the description of model, the assumptions and problems to be addressed.
Section \ref{section_theory} introduces the theory of discrete-cISSIM.
Section \ref{section_algorithm} introduces the identification algorithms in detail.
Section \ref{section_consistency} gives the consistency proofs of above algorithms, and \ref{section_simulation} shows the simulation results.
Section \ref{section_conclusion} concludes the paper.

\textbf{Notations}:
The imaginary unit is denoted by $\im$.
Let $x\in\C$, the complex conjugate of $x$ is denoted by $\overline{x}$.

In the following notations, assume $\mat[A] = \begin{bmatrix} \vec[a]_1 & \cdots & \vec[a]_r \end{bmatrix}$, $\mat[B] = \begin{bmatrix} \vec[b]_1 & \cdots & \vec[b]_r \end{bmatrix}$ to be two matrices where $\vec[a]_i$ and $\vec[b]_j$ are all vectors.
Assume $\vec[x]$ as another vector.
The spectral radius is denoted by $\rho(\mat[A])$ if $\mat[A]$ is a square matrix.
The Euclidean norm of vector $\vec[x]$ is denoted by ${\left\Vert\vec[x]\right\Vert}_2$
The spectral and Frobenius norm of matrix $\mat[A]$ are denoted by ${\left\Vert\mat[A]\right\Vert}_2$ and ${\left\Vert\mat[A]\right\Vert}_F$ respectively.
The element-wise matrix multiplication is denoted by $\mat[A] .* \mat[B]$.
The Moore-Penrose inverse is denoted by $\mat[A]^\dagger$.
The Kronecker product of matrices $\mat[A]$ and $\mat[B]$ is denoted by $\mat[A] \otimes \mat[B]$
The Khatri-Rao product of matrices $\mat[A]$ and $\mat[B]$ is denoted by $\mat[A]\odot\mat[B] = \begin{Bmatrix} \vec[a]_1\otimes\vec[b]_1 & \cdots & \vec[a]_r\otimes\vec[b]_r \end{Bmatrix}$.

\section{Problem Formation} \label{section_problem}
Suppose the true plant is captured by the following DT LTI state-space model:
\begin{align}
    \vec[x]_{k+1} & = \mat[A]\vec[x]_k + \mat[B]\vec[u]_k + \vec[\omega]_k, \label{eq_problem_model_state} \\
    \vec[y]^{\rm m}_k & = \mat[C]\vec[x]_k + \mat[D]\vec[u]_k + \vec[\nu]_k, \label{eq_problem_model_output} \\
    \vec[u]^{\rm m}_k & = \vec[u]_k + \vec[\tau]_k, \label{eq_problem_model_input}
\end{align}
where $\vec[x]_k \in \R^n$ is the system state, $\vec[u]^{\rm m}_k \in \R^{\rm m}$ and $\vec[y]^{\rm m}_k \in \R^p$ are the measured input and output, respectively.
The noises $\vec[\omega]_k$, $\vec[\nu]_k$, and $\vec[\tau]_k$ are all assumed as zero-mean Gaussian white noise with the covariance matrices defined as follows:
\begin{equation} \label{eq_problem_covariance}
    \E\left[\begin{pmatrix*} \vec[\omega]_k \\ \vec[\nu]_k \\ \vec[\tau]_k \end{pmatrix*}
        \begin{pmatrix*} \vec[\omega]\tr_k & \vec[\nu]\tr_k & \vec[\tau]\tr_k \end{pmatrix*}\right] =
    \begin{bmatrix*}
        \mat[\Sigma]_{\omega\omega} & \mat[0] & \mat[0] \\ \mat[0] & \mat[\Sigma]_{\nu\nu} & \mat[0] \\ \mat[0] & \mat[0] & \mat[\Sigma]_{\tau\tau}
    \end{bmatrix*},
\end{equation}
where $\Sigma_{\omega\omega}\in\R^{n\times n}$, $\Sigma_{\nu\nu}\in\R^{p\times p}$ and $\Sigma_{\tau\tau}\in\R^{m\times m}$ are all semi-positive definite matrices.

The following assumptions are essential for derivation:

\begin{assumption} \label{assume_stability}
    $\mat[A]$ is asymptotically (Schur) stable.
\end{assumption}

\begin{assumption} \label{assume_minimum}
    $(\mat[A], \mat[B])$ is controllable, and $(\mat[C], \mat[A])$ is observable.
\end{assumption}

\begin{assumption} \label{assume_order_upperbound}
    The upper bound of the system order $\bar{n}$ is given as \emph{a priori}.
\end{assumption}

\begin{assumption} \label{assume_discrete_periodic}
    The excitation $\vec[u]_k$ is discretely periodic: $\vec[u]_{k+T} = \vec[u]_k$ where the minimal period $T\in\mathbb N^+$ is known.
\end{assumption}

The identification problem of the paper is divided into the deterministic and stochastic parts as follows:

\begin{problem}[system parameters identification] \label{problem_deterministic_identification}
    Given the noise corrupted input-output samples $\{\vec[u]_k^{\rm m}\}$, $\{\vec[y]_k^{\rm m}\}$, $k = 1,2,\ldots,N$, the problem is to find \emph{both of the followings} as samples number $N \to \infty$,
    \begin{itemize}
        \item the system order $n$;
        \item the consistent estimation of the quadruple $(\mat[A], \mat[B], \mat[C], \mat[D])$ up to a similarity transformation.
    \end{itemize}
\end{problem}

\begin{problem}[covariance matrices identification] \label{problem_stochastic_identification}
    Given the same data as in Problem \ref{problem_deterministic_identification}, the problem is to find the triple $\left(\mat[\Sigma]_{\omega\omega},\mat[\Sigma]_{\nu\nu},\mat[\Sigma]_{\tau\tau}\right)$ such that \emph{both of the followings} up to a similarity transformation as $N\to\infty$.
\end{problem}

\section{Discrete Invariant Subspace Theory} \label{section_theory}
The invariant subspace theory was made to build up the relationship between the time-domain samples and frequency-domain parameters.
This section will propose a modification to the theory so that the excitation can be arbitrary periodical signal in discrete-time domain.

\subsection{Excitation System}
The excitation system is a \textbf{virtual} dynamic system for modeling input (excitation) signal.
First, we define the \textit{components} of the excitation system as follows:
\begin{align}
    \widetilde{\mat[S]}_r & = \label{eq_element_mat_s}
    \begin{cases*}
        1, & $r = 0$; \\
        \printLongest{\begin{bmatrix*} \cos \omega r & \sin \omega r \\ -\!\sin \omega r & \cos \omega r \end{bmatrix*},} & $r = 1,\ldots,\lfloor{{T}\over{2}}\rfloor$;
    \end{cases*} \\
    \widetilde{\mat[U]}_r & = \label{eq_element_mat_u}
    \begin{cases*}
        \continueToAlign{\vec[b]_r,} & $r = 0$; \\
        \begin{bmatrix*} \vec[a]_r & \vec[b]_r \end{bmatrix*}, & $r = 1,\ldots,\lfloor{{T}\over{2}}\rfloor$;
    \end{cases*} \\
    \widetilde{\vec[v]}_{0,r} & = \label{eq_element_vec_v0} \begin{cases*}
        \continueToAlign{1/\sqrt{2},} & $r = 0$; \\
        \begin{bmatrix*} \sin\phi_r & \cos\phi_r \end{bmatrix*}, & $r = 1,\ldots,\lfloor{{T}\over{2}}\rfloor$,
                    \end{cases*}
\end{align}
where $\vec[a]_r\in\R^m$, $\vec[b]_r\in\R^m$ are constant vectors and $\omega = \frac{2\pi}{T}$.

Then $\vec[u]_k$ which satisfies Assumption \ref{assume_discrete_periodic} can be generated by the full excitation system ($\rm f$ for full) below:
\begin{equation} \begin{aligned} \label{eq_excitation_system}
        \vec[v]_{k+1}^{\rm f} & = \mat[S]^{\rm f}\vec[v]_{k}^{\rm f}, \\
        \vec[u]_k & = \mat[U]^{\rm f}\vec[v]_{k}^{\rm f},
    \end{aligned} \end{equation}
where $\mat[S]^{\rm f}$, $\mat[U]^{\rm f}$ and $\vec[v]_0^{\rm f}$ are defined as follows:
\begin{align}
    \mat[S]^{\rm f} & = \op[block diag]\left(\widetilde{\mat[S]}_0, \widetilde{\mat[S]}_1, \ldots, \widetilde{\mat[S]}_{\lfloor{{T}\over{2}}\rfloor} \right), \label{eq_excitation_system_matrix} \\
    \mat[U]^{\rm f} & = \begin{bmatrix*} \widetilde{\mat[U]}_0 & \widetilde{\mat[U]}_1 & \cdots & \widetilde{\mat[U]}_{\lfloor{{T}\over{2}}\rfloor} \end{bmatrix*}, \label{eq_excitation_output_matrix} \\
    \vec[v]_0^{\rm f} & = \begin{bmatrix*} \widetilde{\vec[v]}_{0,1} & \widetilde{\vec[v]}_{0,1} & \cdots & \widetilde{\vec[v]}_{0,\lfloor{{T}\over{2}}\rfloor} \end{bmatrix*}\tr. \label{eq_excitation_init_state}
\end{align}
Equations \eqref{eq_excitation_system}-\eqref{eq_excitation_init_state} define a neutrally-stable DT LTI system generating multi-sine signals.
According to discrete Fourier transform, this system can be used to represent arbitrary discrete periodic signals.

\begin{remark} \label{remark_fourier_coefficient}
        Full excitation system \eqref{eq_excitation_system}-\eqref{eq_excitation_init_state} is the frequency expression of $\vec[u]_k$.
    Assume $\phi_r = 0$ for $r > 1$.
        For example, when $T$ is \textbf{odd}, the Fourier coefficient denoted by $\widetilde{\mathcal{U}}_r$ satisfies that $\widetilde{\mathcal{U}}_0 = {1\over\sqrt{2}}\vec[b]_0$, $\widetilde{\mathcal{U}}_r = {1\over2}(\vec[b]_r-\im\vec[a]_r)$ for $r=1,\ldots,\lfloor{{T}\over{2}}\rfloor$, such that the following equation holds:
    \begin{equation*}
        \vec[u]_k = \sum_{r=0}^{T-1}{\widetilde{\mathcal{U}}_r e^{\im k\omega r}} = \vec[b]_0 + \sum_{r=1}^{\lfloor{{T}\over{2}}\rfloor}{\vec[a]_r\sin(k\omega r)+\vec[b]_r\cos(k\omega r)},
    \end{equation*}
        As $\vec[u]_k$ is real, $\widetilde{\mathcal{U}}_r = \overline{\widetilde{\mathcal{U}}_{N-r \mod{N}}}$ for $r=0, \ldots, T-1$.
        It means that there are $\lfloor{{T}\over{2}}\rfloor + 1$ \textit{independent} frequencies.
        It is also discussed in Lemma \ref{lemma_orthogonal}.
\end{remark}

\subsection{Invariant Subspace Equation}
The full excitation system \eqref{eq_excitation_system}-\eqref{eq_excitation_init_state} can describe $\vec[u]_k$ rigorously.
However, some frequencies are not useful, even detrimental in identification.
Discrete-time domain allows to choose an arbitrary set of the excitation frequencies set to use in identification:
Define the \textbf{overall excitation frequency set} as $\mathbb{F}^{\rm f} = \left\{ 0, \omega, \ldots, \omega\lfloor{{T}\over{2}}\rfloor \right\}$ ($\rm f$ for full),
the \textbf{interested excitation frequency set} as $\mathbb{F}^{\rm i} = \left\{ \omega\indk_1, \ldots, \omega\indk_q \right\} \subset \mathbb{F}^{\rm f}$ ($\rm i$ for interested),
and the \textbf{complement set} as $\mathbb{F}^{\rm c} = \mathbb{F}^{\rm f} \backslash \mathbb{F}^{\rm i} = \left\{ \omega\indl_1, \ldots, \omega\indl_{\lfloor{{T}\over{2}}\rfloor-q+1} \right\}$ ($\rm c$ for complement).
Thus, the divided invariant subspace equation (ISE) is proposed as follows:

\begin{theorem} \label{theorem_ise}
        Suppose that true plant system \eqref{eq_problem_model_state}-\eqref{eq_problem_model_input} satisfies Assumption \ref{assume_stability}.
    Excite the system with the signal $\vec[u]_k$ which satisfies Assumption \ref{assume_discrete_periodic}.
        Then the state $\vec[x]_k$, the measured output $\vec[y]_k^{\rm m}$ and measured input $\vec[y]_k^{\rm m}$ satisfy the following regressions:
    \begin{align}
        \vec[x]_k & = \mat[X]^{\rm i}\vec[v]_k^{\rm i} + \mat[X]^{\rm c}\vec[v]_k^{\rm c} + \vec[z]_k, \label{eq_state_regression}    \\
        \vec[y]_k^{\rm m} & = \mat[Y]^{\rm i}\vec[v]_k^{\rm i} + \mat[Y]^{\rm c}\vec[v]_k^{\rm c} + \vec[r]_k, \label{eq_output_regression}   \\
        \vec[u]_k^{\rm m} & = \mat[U]^{\rm i}\vec[v]_k^{\rm i} + \mat[U]^{\rm c}\vec[v]_k^{\rm c} + \vec[\tau]_k. \label{eq_input_regression}
    \end{align}
    where $\mat[X]^{\rm i}$, $\mat[X]^{\rm c}$, $\mat[Y]^{\rm i}$, $\mat[Y]^{\rm c}$, $\mat[U]^{\rm i}$ and $\mat[U]^{\rm c}$ satisfy the following Sylvester equations called ISE:
    \begin{align}
        \mat[X]^{\rm i}\mat[S]^{\rm i} + \mat[X]^{\rm c}\mat[S]^{\rm c}
            & = \mat[A](\mat[X]^{\rm i} + \mat[X]^{\rm c}) + \mat[B](\mat[U]^{\rm i}   + \mat[U]^{\rm c}), \label{eq_ise_state}  \\
        \mat[Y]^{\rm i} + \mat[Y]^{\rm c}
            & = \mat[C](\mat[X]^{\rm i} + \mat[X]^{\rm c}) + \mat[D](\mat[U]^{\rm i} + \mat[U]^{\rm c}); \label{eq_ise_output}
    \end{align}
    while $\vec[z]_k$ and $\vec[r]_k$ satisfy the following disturbance system:
    \begin{align}
        \vec[z]_{k+1} & = \mat[A]\vec[z]_k + \vec[\omega]_k, \label{eq_disturbance_state}                                              \\
        \vec[r]_k     & = \mat[C]\vec[z]_k + \vec[\nu]_k, \label{eq_disturbance_output}                                                \\
        \vec[z]_0     & = \vec[x]_0 - \mat[X]^{\rm i}\vec[v]_0^{\rm i} - \mat[X]^{\rm c}\vec[v]_0^{\rm c}, \label{eq_disturbance_init}
    \end{align}
    where $\vec[z]_k = \vec[x]_k - \mat[X]^{\rm i}\vec[v]_k^{\rm i} - \mat[X]^{\rm c}\vec[v]_k^{\rm c}$ is the state disturbance, while $\vec[r]_k = \vec[y]_k - \mat[Y]^{\rm i}\vec[v]_k^{\rm i} - \mat[Y]^{\rm c}\vec[v]_k^{\rm c}$ is the output disturbance.
\end{theorem}

\begin{proof}
        From Section 4 in \cite{cISSIM_CHuang_1}, the true plant model can be divided into the \textit{integrated} ISE and disturbance system:
    \begin{equation} \begin{aligned} \label{eq_iseproof_subsystem}
            \mat[X]^{\rm f}\mat[S]^{\rm f}
                & = \mat[A]\mat[X]^{\rm f} + \mat[B]\mat[U]^{\rm f},
                & \quad \vec[z]_{k+1}
                & = \mat[A]\vec[z]_k + \vec[\omega]_k, \\
            \mat[Y]^{\rm f}
                & = \mat[C]\mat[X]^{\rm f} + \mat[D]\mat[U]^{\rm f},
                & \quad \vec[r]_k
                & = \mat[C]\vec[z]_k + \vec[\nu]_k, \\
            &   & \quad \vec[z]_0
                & = \vec[x]_0 -\mat[X]^{\rm f}\vec[v]_0^{\rm f};
        \end{aligned} \end{equation}
    and the following equations holds from definition:
    \begin{equation} \begin{aligned} \label{eq_iseproof_definition}
            \vec[x]_k         & = \mat[X]^{\rm f}\vec[v]_k^{\rm f} + \vec[z]_k,
                              & \quad \vec[u]_k^{\rm m}                         & = \mat[U]^{\rm f}\vec[v]_k^{\rm f} + \vec[\tau]_k, \\
            \vec[y]_k^{\rm m} & = \mat[Y]^{\rm f}\vec[v]_k^{\rm f} + \vec[r]_k.                                                      \\
        \end{aligned} \end{equation}
        The existence and uniqueness of the solution to ISE $\{\mat[X]^{\rm f}, \mat[Y]^{\rm f}\}$ hold true because $\mat[S]^{\rm f}$ and $\mat[A]$ have different eigenvalues \cite[Appendix A]{OutputRegulation_Huang_1}.

        The matrices and vector in the above two subsystems can be divided equivalently according to the frequency sets $\mathbb{F}^{\rm i}$ and $\mathbb{F}^{\rm c}$:
    \begin{itemize}
        \item $\mat[S]^{\rm f} \Rightarrow \op[block diag]\left( \mat[S]^{\rm i}, \mat[S]^{\rm c} \right)$, and $\vec[v]_{k}^{\rm f} \Rightarrow \begin{bmatrix*} (\vec[v]_{k}^{\rm i})\tr & (\vec[v]_{k}^{\rm c})\tr \end{bmatrix*}\tr$;
        \item $\mat[U]^{\rm f} \Rightarrow \begin{bmatrix*} \mat[U]^{\rm i} & \mat[U]^{\rm c} \end{bmatrix*}$, and so do the $\mat[X]^{\rm f}$, $\mat[Y]^{\rm f}$.
    \end{itemize}
    The division is valid because the definitions and orders of their \textit{components} are independent in the above matrices and vector.
        Then the integrated subsystems \eqref{eq_iseproof_subsystem}-\eqref{eq_iseproof_definition} can be rewritten into decomposed version \eqref{eq_ise_state}-\eqref{eq_disturbance_init} using block matrix multiplication.
    That completes the proof.
\end{proof}

\begin{remark}[stationary auto-correlation] \label{remark_ise_stationary_process}
        The disturbance system \eqref{eq_disturbance_state}-\eqref{eq_disturbance_init} describes both the transient and noises process in identification:
    The transient process is described by equation \eqref{eq_disturbance_init}, while the noise is described by a DT LTI autoregressive moving average (ARMA) process driven by a zero-mean Gaussian noises.
        The disturbance system is asymptotically stationary as $N\to\infty$ since the system is stable (Assumption \ref{assume_stability}).

        When the system has reached the steady state, it is an ergodic process, which means the second moments are time invariant \cite{StochasticDifferentialEquation_Gard_1}.
        Define the auto-correlation matrices as follows:
    \begin{equation} \begin{aligned} \label{eq_correlation_definition}
            \mat[\Xi]_{zz} & = \E(\vec[z]_k\vec[z]\tr_k), \quad \mat[\Xi]_{tt} = \E(\vec[\tau]_k\vec[\tau]\tr_k), \\
            \mat[\Xi]_{rr}^\indl & = \E(\vec[r]_{k+\indl}\vec[r]\tr_k), \indl = 0, 1, \ldots. \\
        \end{aligned} \end{equation}
    In this case, all the auto-correlation matrices are independent with $k$, allowing them to be computed as follows:
    \begin{align}
        \mat[\Xi]_{zz} & = \mat[A]\mat[\Xi]_{zz}\mat[A]\tr + \mat[\Sigma]_{\omega\omega}, \label{eq_correlation_zz} \\
        \mat[\Xi]_{rr}^\indl & =
        \begin{cases*}
            \mat[C]\mat[\Xi]_{zz}\mat[C]\tr + \mat[\Sigma]_{\nu\nu}, & $\indl = 0$; \\
            \mat[C]\mat[A]^\indl\mat[\Xi]_{zz}\mat[C]\tr, & $\indl = 1,2,\ldots$.
        \end{cases*} \label{eq_correlation_rr}
    \end{align}
\end{remark}

\section{Identification Algorithms} \label{section_algorithm}
This section includes three parts:
Sections \ref{section_isim_identification} and \ref{section_sim_identification} solve the Problem \ref{problem_deterministic_identification}; Section \ref{section_noise_identification} solves the problem \ref{problem_stochastic_identification} along with the precedent sections; and Section \ref{section_algorithm_conclusion} summarizes them.

\subsection{Invariant Subspace Identification} \label{section_isim_identification}
This subsection uses the linear regression models \eqref{eq_output_regression}-\eqref{eq_input_regression} to estimate the invariant subspace coefficients $\mat[Y]^{\rm i}$ and $\mat[U]^{\rm i}$ selected in $\mathbb{F}^{\rm i}$.
Since the two models use the same regressor, they are written into one regression model using the following notations:
\begin{equation*} \begin{aligned}
    & (\vec[z]_k^{\rm m})\tr = \begin{bmatrix*} (\vec[y]_0^{\rm m})\tr & (\vec[u]_0^{\rm m})\tr \end{bmatrix*}, \\
    & \mat[Z]_N = \begin{bmatrix*} (\vec[z]_0^{\rm m})\tr \\ \vdots \\ (\vec[z]_{N-1}^{\rm m})\tr \end{bmatrix*},
        \mat[E]_N = \begin{bmatrix*} \vec[r]\tr_0 & \vec[\tau]\tr_0 \\ \vdots & \vdots \\ \vec[r]_{N-1}\tr & \vec[\tau]_{N-1}\tr \end{bmatrix*}, \\
    & \mat[F]_N^{\rm i} = \begin{bmatrix*} \vec[v]_0^{\rm i} & \cdots & \vec[v]_{N-1}^{\rm i} \end{bmatrix*}\tr,
        \mat[R]^{\rm i} = \begin{bmatrix*} (\mat[Y]^{\rm i})\tr & (\mat[U]^{\rm i})\tr \end{bmatrix*},
\end{aligned} \end{equation*}
where the matrices $\mat[F]_N^{\rm c}$ and $\mat[R]^{\rm c}$ for complement frequencies are defined similarly.
The linear regression equation can be written as follows:
\begin{equation} \label{eq_isim_ols_model}
    \mat[Z]_N = \mat[F]_N^{\rm i}\mat[R]^{\rm i} + \mat[F]_N^{\rm c}\mat[R]^{\rm c} + \mat[E]_N.
\end{equation}

In the \textit{trivial} case that $\mathbb{F}^{\rm c}$ is null, the model is degenerated into $\mat[Z]_N = \mat[F]_N^{\rm i}\mat[R]^{\rm i} + \mat[E]_N$.
To minimize the quadratic error loss, the pseudo-inverse based least-square method is used in the precedent cISSIM \cite{cISSIM_CHuang_2}:
ordinary least-squares (OLS) when offline: $\hat{\mat[R]}^{\rm i} = (\mat[F]_N^{\rm i})^\dagger \mat[Z]_N$; and recursive least-squares (RLS) when online (see \cite[Algorithm 1]{cISSIM_CHuang_2}).
It takes $\mathcal{O}(q^2N)$ computational operations.

In the DT domain, model \eqref{eq_isim_ols_model} can be handled within only $\mathcal{O}(qN)$ or $\mathcal{O}(N\log T)$ operations, regardless of $\mathbb{F}^{\rm c}$ is null or not.
The only requirement is that the sample number $N = kT$.
Propose the following orthogonal lemma first:

\begin{lemma} \label{lemma_orthogonal}
    Define $\mat[F]_N^{\rm i}$ and $\mat[F]_N^{\rm c}$ as above,
    When $N = kT$, $k = 1, 2, \ldots$, it holds that,
    \begin{equation} \label{eq_lemma_orthogonal}
        (\mat[F]_{kT}^{\rm f})\tr\mat[F]_{kT}^{\rm f} =
        \begin{cases*}
            \frac{kT}{2}\mat[I]_{2\lfloor{{T}\over{2}}\rfloor+1}, & $T$ is odd; \\
            \frac{kT}{2}\op[blkdiag]\left( \mat[I]_{2\lfloor{{T}\over{2}}\rfloor}, 2\mat[H] \right), & $T$ is even,
        \end{cases*}
    \end{equation}
    where $H = \begin{bmatrix*}
            \sin^2\!\phi_{{{T}\over{2}}} & \cos\!\phi_{{{T}\over{2}}}\sin\!\phi_{{{T}\over{2}}} \\
            \cos\!\phi_{{{T}\over{2}}}\sin\!\phi_{{{T}\over{2}}} & \cos^2\!\phi_{{{T}\over{2}}}
        \end{bmatrix*}$.
\end{lemma}
\begin{proof} See \ref{appendix_orthogonal_projection}. \end{proof}

\begin{corollary} \label{corollary_orthogonal}
    When $N = kT$, and $r = {{{T}\over{2}}}$ is excluded from $\mathbb{F}^{\rm f}$ if $T$ is even, then there is:
    \begin{equation} \label{eq_orthogonal_corollary}
        (\mat[F]_{kT}^{\rm i})\tr\mat[F]_{kT}^{\rm i} = \frac{kT}{2}\mat[I]_s, \quad
        (\mat[F]_{kT}^{\rm i})\tr\mat[F]_{kT}^{\rm c} = \mat[0].
    \end{equation}
\end{corollary}

When Corollary \ref{corollary_orthogonal} holds, $\mat[F]_N^{\rm i}$ is of full column rank and the pseudo inverse can be simplified: $(\mat[F]_N^{\rm i})^\dagger = \frac{2}{kT}(\mat[F]_N^{\rm i})\tr$.
Hence, the first residual term $\mat[F]_N^{\rm c}\mat[R]^{\rm c} = \mat[0]$ when $N = kT$ or $N\to\infty$, and the second residual term $\mat[E]_N$ is asymptotically disappeared when $N\to\infty$ (see Lemma \ref{lemma_isim_consistency} for details).
Thus, the offline identification can be rewritten using \textbf{matrix multiplication} as Algorithm \ref{algorithm_isim_ols}.
The matrix multiplication is partitioned for better cache behavior during implementation \cite{Numerical_Cormen_1}.

\begin{algorithm}[htb]
    \LinesNumbered \SetAlgoLined
    \caption{Offline ISP by Matrix Multiplication} \label{algorithm_isim_ols}
    \KwData{ $\mat[Z]_N$, $\mat[F]_N$. }
    \KwResult{ $\hat{\mat[R]}^{\rm i}$. }

    Initialize $\hat{\mat[R]}^{\rm i} = \mat[0]$\;
    \For{$k = 1, \ldots, \lfloor{{N}\over{T}}\rfloor$}{
    $\hat{\mat[R]}^{\rm i} \leftarrow \hat{\mat[R]}^{\rm i} + \left(\lfloor{{N}\over{T}}\rfloor {T\over 2}\right)^{-1} \times (\mat[F]_T^{\rm i})\tr (\mat[Z]_N)_{\textrm{the $k^\text{th}$ $T$ rows}}$\;
    }
    \Return $\hat{\mat[R]}^{\rm i}$.
\end{algorithm}

The related online identification is given as Algorithm \ref{algorithm_isim_rls} shown.
Note that it outputs only when $N = kT$, $k\in\N^+$.

\begin{algorithm}[htb]
    \LinesNumbered \SetAlgoLined
    \caption{Online ISP by Matrix Multiplication} \label{algorithm_isim_rls}
    \KwData{ $\{(\vec[z]_k^{\rm m})\tr\}$, $k = 0,1,2,\ldots$. }
    \KwResult{ $\hat{\mat[R]}^{\rm i}$. }
    \tcc{Assume $\indk_1 = 0$ in $\mathbb{F}^{\rm i} = \left\{ \omega\indk_1, \ldots, \omega\indk_q \right\}$}

    Initialize $\hat{\mat[R]}^{\rm i}, \hat{\mathbcal{W}}_{\rm on} = \mat[0]$, $k = 0$, $t = 0$\;
    \While{true}{
        $k \leftarrow k + 1$\;
        $\vec[w]_k\tr \leftarrow \begin{bmatrix*} 1 & \sin(k\omega\indk_1+\phi_{\indk_1}) & \cos(k\omega\indk_1+\phi_{\indk_1}) & \cdots \end{bmatrix*}$\;
        $\hat{\mathbcal{W}}_{\rm on} \leftarrow \hat{\mathbcal{W}}_{\rm on} + \vec[w]_k (\vec[z]_k^{\rm m})\tr$\;
        \If{$k \mod T = 0$}{
            $t \leftarrow t + 1$\;
            $\hat{\mathbcal{W}}_{\rm on} \leftarrow {2\over T} \times \hat{\mathbcal{W}}_{\rm on}$\;
            \textbf{update} $\hat{\mat[R]}_t^{\rm i} \leftarrow {{t-1}\over{t}} \times \hat{\mat[R]}_{t-1}^{\rm i} + {{1}\over{t}} \times \hat{\mathbcal{W}}_{\rm on}$\;
            $\hat{\mathbcal{W}}_{\rm on} \leftarrow \mat[0]_{s\times(p+m)}$.
        }
    }
\end{algorithm}

In fact, Algorithms \ref{algorithm_isim_ols} and \ref{algorithm_isim_rls} are the variants of \textbf{DFT}, which ensures the ability of eliminating the aliasing and unwanted frequencies given $N = kT$.
Therefore, when $q \gg 1$, it can be accelerated by using \textbf{fast Fourier transformation} (FFT), which consumes $\mathcal{O}(N\log T)$ operations.
Algorithm \ref{algorithm_isim_fft} shows an example for $q = \lfloor{{T+1}\over{2}}\rfloor$.

\begin{algorithm}[htb]
    \LinesNumbered \SetAlgoLined
    \caption{Offline ISP by FFT when $q = \lfloor{{T+1}\over{2}}\rfloor$} \label{algorithm_isim_fft}
    \KwData{ $\mat[Z]_N$, $\mat[F]_N$. }
    \KwResult{ $\hat{\mat[R]}^{\rm i}$. }

    Initialize $\hat{\mathbcal{W}}_{\rm fft} = \mat[0]$\;
    \For{$k = 1, \ldots, \lfloor{{N}\over{T}}\rfloor$}{
        \tcc{Do FFT column-wisely}
        $\hat{\mathbcal{W}}_{\rm fft} \leftarrow \hat{\mathbcal{W}}_{\rm fft} + \left(\lfloor{{N}\over{T}}\rfloor T\right)^{-1} \times \op[fft]\left[ (\mat[Z]_N)_{\textrm{the $k^\text{th}$ $T$ rows}} \right]$\;
    }
    $\hat{\mathbcal{A}}_{\rm fft} \leftarrow -2\op[imag]\left[(\hat{\mathbcal{W}}_{\rm fft})_\textrm{$2{\sim}\lfloor{{T}\over{2}}\rfloor+1$}^\textrm{columns at}\right]$\;
    $\hat{\mathbcal{B}}_{\rm fft} \leftarrow \begin{bmatrix*}
        \sqrt{2}(\hat{\mathbcal{W}}_{\rm fft})_\textrm{the $1^\text{st}$}^\textrm{column} &  
        2\op[real]\left[(\hat{\mathbcal{W}}_{\rm fft})_\textrm{$2{\sim}\lfloor{{T}\over{2}}\rfloor+1$}^\textrm{columns at}\right]
    \end{bmatrix*}$\;
    \Return $\hat{\mat[R]}^{\rm i} \leftarrow \textrm{interlace along the rows}\left(\hat{\mathbcal{B}}_{\rm fft}, \hat{\mathbcal{A}}_{\rm fft}\right)$.
\end{algorithm}

The two methods proposed here have different applicable scenarios: \textit{matrix-multiplication} is suitable for smaller $q$; while \textit{FFT} is faster when $q$ is close to $T$ such as covariance estimation in Section \ref{section_noise_identification}.
The consistency of both two methods are ensured since they use the DFT matrix as regressors, which is unitary.

\subsection{Subspace Identification} \label{section_sim_identification}
Referring to the cISSIM \cite{cISSIM_CHuang_2}, the order $n$ and the quadruple of the system parameters $(\mat[A], \mat[B], \mat[C], \mat[D])$ can be obtained by subspace identification.
It uses $\mat[Y]^{\rm i}$ and $\mat[U]^{\rm i}$ as inputs, and is divided into two parts:

\subsubsection{Estimation of A, C and n} \label{section_sim_acn}
Considering the Sylvester equations \eqref{eq_ise_state}-\eqref{eq_ise_output}, $\mat[A]$, $\mat[C]$ and $n$ can be estimated in the similar way as in continuous-time domain \cite[Lemma 5]{cISSIM_CHuang_2}.
The consistency of the following Algorithm refers to \cite[Lemma 6]{cISSIM_CHuang_2}, where $\op[sqrt](\mat[M])$ stands for making square root onto matrix $\mat[M]$ element-wisely.
Algorithm \ref{algorithm_sim_acn} shows the results.

\begin{algorithm}[htb]
    \LinesNumbered \SetAlgoLined
    \caption{Estimate $\mat[A]$, $\mat[C]$ and $n$} \label{algorithm_sim_acn}
    \KwData{ $\hat{\mat[Y]}^{\rm i}$, $\hat{\mat[U]}^{\rm i}$, $\bar{n}$. }
    \KwResult{ $\hat{\mat[A]}$, $\hat{\mat[C]}$ and $\hat{n}$. }

    Build $\hat{\mathbcal{Y}}$ and $\hat{\mathbcal{U}}$:\\
    $\mathbcal{Y} = \begin{bmatrix*} \mat[Y]^{\rm i} \\ \mat[Y]^{\rm i}\mat[S]^{\rm i} \\ \vdots \\ \mat[Y]^{\rm i}(\mat[S]^{\rm i})^{\bar{n}-1} \end{bmatrix*}$,
    $\mathbcal{U} = \begin{bmatrix*} \mat[U]^{\rm i} \\ \mat[U]^{\rm i}\mat[S]^{\rm i} \\ \vdots \\ \mat[U]^{\rm i}(\mat[S]^{\rm i})^{\bar{n}-1} \end{bmatrix*}$\;

    QR decomposition: $\begin{bmatrix*} \hat{\mathbcal{U}}\tr & \hat{\mathbcal{Y}}\tr \end{bmatrix*} =
        \begin{bmatrix*} \hat{\mat[Q]}_1 & \hat{\mat[Q]}_2 \end{bmatrix*}
        \begin{bmatrix*} \hat{\mat[R]}_{1,1} & \hat{\mat[R]}_{1,2} \\ \mat[0] & \hat{\mat[R]}_{2,2} \end{bmatrix*}$\;
    Slice $\hat{\mat[R]}_{2,2} \in \R^{\bar{n}p\times\bar{n}p}$ and conduct SVD: $\hat{\mat[R]}_{2,2}\tr
        = \hat{\mat[W]} \hat{\mat[\Sigma]} \hat{\mat[Z]}\H
        = \begin{bmatrix*} \hat{\mat[W]}_1 & \hat{\mat[W]}_0 \end{bmatrix*} \begin{bmatrix*} \hat{\mat[\Sigma]}_1 & \mat[0] \\ \mat[0] & \hat{\mat[\Sigma]}_0 \end{bmatrix*} \hat{\mat[Z]}\H$\;
    Select the largest $\hat{n}$ singular values that forms $\hat{\mat[\Sigma]}_1$\;
    $\hat{\mat[O]} \leftarrow \hat{\mat[W]}_1 \; \op[sqrt](\hat{\mat[\Sigma]}_1)$\;
    $\hat{\mat[C]} \leftarrow \hat{\mat[O]}_{\textrm{the first $p$ rows}} $\;
    $\hat{\mat[A]} \leftarrow \left( \hat{\mat[O]}_{\textrm{the first $p(\bar{n}-1)$ rows}} \right)^\dagger \hat{\mat[O]}_{\textrm{the last $p(\bar{n}-1)$ rows}}$\;
    \Return $\hat{\mat[A]}$, $\hat{\mat[C]}$ and $\hat{n}$.
\end{algorithm}

\begin{remark}[order selection] \label{remark_order_selection}
        The system order $n$ is determined by Step $4$ in Algorithm \ref{algorithm_sim_acn}.
    Frequency-domain SIM \cite{SID_Overschee_1} has proven that non-zero singular values asymptotically converge into the system order.
        Thus, a rough estimation can be made through calculating the rank of $\hat{\mat[R]}_{2, 2}$ by SVD with an adjustable tolerance \cite{OrderSelection_Roy_1}.
        Readers can refer to literature \cite{SystemIdentification_Ljung_1} for more practical ways, such as Akaike's Information Criterion (AIC) and Akaike's Final Prediction Error (FPE).
    The detailed discussion is beyond the scope of this paper.
        Moreover, selecting the frequencies that have high amplitudes into $\mathbb{F}^{\rm i}$ can also help to increasing the signal noise ratios (SNR), and improve the estimation.
\end{remark}

\subsubsection{Estimation of B and D} \label{section_sim_bdx}
Similar to cISSIM \cite{cISSIM_CHuang_2}, $\mat[B]$ and $\mat[D]$ can be estimated through decomposing the Sylvester equations \eqref{eq_ise_state}-\eqref{eq_ise_output} into multivariate linear regression problem, but in a different decomposition:

$\mat[S]$ is diagonalizable because of its skew Hermitianity.
Assume that the first term of $\mathbb{F}^{\rm i}$ is null frequency: $\indk_1 = 0$, then the diagonalizing and diagonalized matrix of $\mat[S]$ are defined as follows such that $\widetilde{\mat[P]}^{-1}\mat[S]^{\rm i}\widetilde{\mat[P]} = \mat[\Omega]$:
\begin{align}
    \widetilde{\mat[P]} & = \op[blockdiag]\left( 1, \begin{bmatrix*} \im & -\im \\ 1 & 1 \end{bmatrix*}, \cdots, \begin{bmatrix*} \im & -\im \\ 1 & 1 \end{bmatrix*} \right) \label{eq_sim_diagonalize}, \\
    \mat[\Omega]        & = \op[diag](1, e^{-\im\omega\indk_2}, e^{\im\omega\indk_2}, \ldots, e^{-\im\omega\indk_q}, e^{\im\omega\indk_q}) \label{eq_sim_diagonal}.
\end{align}

Define the similar transformation as $\tilde{\mat[X]}^{\rm i} = \mat[X]^{\rm i} \widetilde{\mat[P]}$, $\tilde{\mat[U]}^{\rm i} = \mat[U]^{\rm i} \widetilde{\mat[P]}$, and $\tilde{\mat[Y]}^{\rm i} = \mat[Y]^{\rm i} \widetilde{\mat[P]}$.
For a specified frequency $\omega\indk_r$, the Sylvester equations can be decomposed as follows:
\begin{equation*}
    e^{\im\omega\indk_r} \tilde{\vec[x]}_{\indk_r} = \mat[A]\tilde{\vec[x]}_{\indk_r} + \mat[B]\tilde{\vec[u]}_{\indk_r}, \quad
    \tilde{\vec[y]}_{\indk_r} = \mat[C]\tilde{\vec[x]}_{\indk_r} + \mat[D]\tilde{\vec[u]}_{\indk_r},
\end{equation*}
where $\tilde{\mat[X]}^{\rm i} = \begin{bmatrix*} \tilde{\vec[x]}_0 & \overline{\tilde{\vec[x]}_{\indk_2}} & \tilde{\vec[x]}_{\indk_2} & \cdots & \overline{\tilde{\vec[x]}_{\indk_q}} & \tilde{\vec[x]}_{\indk_q} \end{bmatrix*}$, and so do the $\tilde{\vec[u]}_{\indk_r}$ and $\tilde{\vec[y]}_{\indk_r}$.
Thus, the Sylvester equations can be represented by a transfer function:
\begin{equation} \label{eq_sim_tf_model}
    \vec[y]_{\indk_r}^{\rm i} = \mat[G]_{\indk_r}^{\rm i} \vec[u]_{\indk_r}^{\rm i}, \quad \mat[G]_{\indk_r}^{\rm i} = \mat[C](e^{\im\omega\indk_r}\mat[I]_n - \mat[A])^{-1}\mat[B] + \mat[D].
\end{equation}
The above equation can be vectorized into linear regression model regarding that $\op[vec](\mat[A]\mat[X]\mat[B]) = (\mat[B]\tr\otimes\mat[A])\op[vec](\mat[X])$:
\begin{equation*} \begin{aligned}
        \op[vec](\tilde{\vec[y]}_{\indk_r}) & = \left[(\tilde{\vec[u]}_{\indk_r}^{\rm i})\tr \otimes \mat[C](e^{\im\omega\indk_r}\mat[I]_n - \hat{\mat[A]})^{-1}\right]\op[vec](\mat[B]) \\
                                            & + \left[(\tilde{\vec[u]}_{\indk_r}^{\rm i})\tr \otimes \mat[I]_p\right]\op[vec](\mat[D]).
\end{aligned} \end{equation*}

The whole method is shown by Algorithm \ref{algorithm_sim_bdx}.

\begin{algorithm}[htb]
    \LinesNumbered \SetAlgoLined
    \caption{Estimate $\mat[B]$ and $\mat[D]$ and $\mat[X]^{\rm i}$} \label{algorithm_sim_bdx}
    \KwData{ $\mat[S]^{\rm i}$, $\hat{\mat[Y]}^{\rm i}$, $\hat{\mat[U]}^{\rm i}$, $\hat{\mat[A]}$, and $\hat{\mat[C]}$. }
    \KwResult{ $\hat{\mat[B]}$ and $\hat{\mat[D]}$. }

    $\hat{\tilde{\mat[X]}}^{\rm i} \, \leftarrow \hat{\mat[X]}^{\rm i} \widetilde{\mat[P]}$, $\hat{\tilde{\mat[U]}}^{\rm i} \leftarrow \hat{\mat[U]}^{\rm i} \widetilde{\mat[P]}$, $\hat{\tilde{\mat[Y]}}^{\rm i} \leftarrow \hat{\mat[Y]}^{\rm i} \widetilde{\mat[P]}$\;
    $\hat{\mat[M]}_{\indk_r} \leftarrow (\hat{\tilde{\vec[u]}}_{\indk_r}^{\rm i})\tr \otimes \hat{\mat[C]}(e^{\im\omega\indk_r}\mat[I]_{\hat{n}} - \hat{\mat[A]})^{-1}$, $r = 1,2,\ldots,q$\;
    $\begin{bmatrix*} \op[vec](\hat{\mat[B]}) \\ \op[vec](\hat{\mat[D]}) \end{bmatrix*} \leftarrow
        {\begin{bmatrix*}
            \hat{\mat[M]}_{\indk_1} & (\hat{\tilde{\vec[u]}}_{\indk_1}^{\rm i})\tr \otimes \mat[I]_p \\
            \vdots & \vdots \\
            \hat{\mat[M]}_{\indk_q} & (\hat{\tilde{\vec[u]}}_{\indk_q}^{\rm i})\tr \otimes \mat[I]_p
        \end{bmatrix*}}^\dagger \op[vec]\hat{\mat[Y]}^{\rm i}$\;
        \Return $\hat{\mat[B]}$ and $\hat{\mat[D]}$.
\end{algorithm}

\subsubsection{Persistent-excitation Condition} \label{section_sim_persistent}
The consistency condition of system parameter estimation is similar to cISSIM:

\begin{assumption}[{{\cite[Assumption 5]{cISSIM_CHuang_2}}}] \label{assume_pe_deterministic}
    Define $\mat[V]$ as follows:
    \begin{equation*}
        \mat[V] = \begin{bmatrix}
            1 & 1 & \cdots & 1 \\
            e^{\im\omega} & e^{\im2\omega} & \cdots & e^{\im\lfloor{{T-1}\over{2}}\rfloor\omega} \\
            \ldots & \ldots & \ddots & \ldots \\
            e^{\im(n+\bar{n}-1)\omega} & e^{\im2(n+\bar{n}-1)\omega} & \cdots & e^{\im\lfloor{{T-1}\over{2}}\rfloor(n+\bar{n}-1)\omega}
        \end{bmatrix}.
    \end{equation*}
    Then the excitation signal should satisfy:
    \begin{equation} \label{eq_deterministic_pecondtion}
        \op[rank](\mat[U]^{\rm i}\odot\mat[V]) = m(n + \bar{n}).
    \end{equation}
\end{assumption}
This paper uses $q$ to denote the number of frequencies that are excited by the noise-free excitation $\vec[u]_k$, and used in $\mathbb{F}^{\rm i}$.
Literature \cite{cISSIM_CHuang_2} shows that condition \eqref{eq_deterministic_pecondtion} satisfies with probability $1$ if $q$ satisfies that $q\geq m(n + \bar{n})$.

\subsection{Covariance Estimation} \label{section_noise_identification}
The covariance estimation is based on ALS and its key step is to separate the noises and excitations from the measured data.
Since the excitation $\vec[u]_k$ can be arbitrary discretely periodic signal, the full excitation frequency set $\mathbb{F}^{\rm f}$ is used:

\begin{assumption} \label{assume_full_identification}
    $\mathbb{F}^{\rm i} = \mathbb{F}^{\rm f}$.
\end{assumption}

\subsubsection{Correlation Matrices Estimation} \label{section_noise_correlation}
The noise samples $\{\vec[r]_k\}$ and $\{\vec[\tau]_k\}$ defined in equation \eqref{eq_iseproof_definition} can be estimated consistently since the transient processes are dying out exponentially in $\vec[r]_k$ and $\vec[\tau]_k$.
Given $\mathbb{F}^{\rm i} = \mathbb{F}^{\rm f}$, the correlation matrices can thus be estimated asymptotically and consistently in the time domain as follows:
\begin{align}
    \hat{\vec[r]}_k & = \vec[y]_k^{\rm m} - \hat{\widetilde{\mat[Y]}}^{\rm i} \vec[v]_k^{\rm i}, \quad
        \hat{\mat[\Xi]}_{rr}^\indl = \frac{\sum_{k=1}^{N-\indl}{\hat{\vec[r]}_{k+\indl}\hat{\vec[r]}\tr_k}}{N-\indl}; \label{eq_als_correlation_rr} \\
    \hat{\vec[\tau]}_k & = \vec[u]_k^{\rm m} - \hat{\widetilde{\mat[U]}}^{\rm i} \vec[v]_k^{\rm i}, \quad
        \hat{\mat[\Xi]}_{tt} = \frac{\sum_{k=1}^{N-1}{\hat{\vec[\tau]}_k\hat{\vec[\tau]}\tr_k}}{N}. \label{eq_als_correlation_tt}
\end{align}
where $\indl = 0,\ldots,M-1$, and $M$ is the total number of $\hat{\mat[\Xi]}_{rr}^\indl$ estimated.
Section \ref{section_noise_persistent} discusses how to choose $M$.
Moreover, they can be calculated faster within frequency-domain considering the Wiener-Khinchin theorem:
\begin{theorem}[{{\cite{TimeSeries_Wiener_1}}}] \label{theorem_wiener_khinchin}
    Given a stochastic process $x(k)$ which is wide-sense stationary, then its auto-correlation function $r_{xy}(\tau) = \E[x(k+\tau) \overline{y(k)}]$ has a spectral decomposition given by the power spectral density $S_{xy}(f) = \frac{1}{2\pi}\sum_{\tau=-\infty}^\infty{r_{xy}(\tau)e^{-\im f\tau}}$:
    \begin{equation*}
        r_{xy}(\tau) = \int_{-\pi}^\pi{e^{\im\tau\omega}S_{xy}(f)\mathrm{d}f}.
    \end{equation*}
\end{theorem}

Denote $\hat{\mat[Y]}_{\rm fft}$ as the DFT of $\{\vec[y]_k^{\rm m}\}$, and $\hat{\mat[U]}_{\rm fft}$ as the DFT of $\{\vec[u]_k^{\rm m}\}$.
Since Algorithm \ref{algorithm_isim_fft} generates them ($\hat{\mathbcal{W}}_{\rm fft} = \begin{bmatrix*} \hat{\mat[Y]}_{\rm fft} & \hat{\mat[U]}_{\rm fft} \end{bmatrix*}$), the correlation matrices can be estimated directly as Algorithm \ref{algorithm_correlation} shows.

Compared with the classical ALS methods, Algorithm \ref{algorithm_correlation} does not need a Kalman state observer, and uses Fourier transformation to reduce the computational complexity from $\mathcal{O}(MN^2)$ to $\mathcal{O}(N\log N)$.

\begin{algorithm}[htb]
    \LinesNumbered \SetAlgoLined
    \caption{Estimate Correlation Matrices} \label{algorithm_correlation}
    \KwData{ $\hat{\mat[Y]}_{\rm fft}$, $\hat{\mat[U]}_{\rm fft}$, $\mat[Y]_N$, $\mat[U]_N$. ($\mat[Z]_N = \begin{bmatrix*} \mat[Y]_N & \mat[U]_N \end{bmatrix*}$) }
    \KwResult{ $\hat{\mat[\Xi]}_{rr}^\indl$, $\hat{\mat[\Xi]}_{tt}^\indl$, $\indl = 0,1,\ldots,M-1$. }

    \tcc{Do FFT and IFFT column-wisely}
    \tcc{Interpolate zeros into frequency axis}
    $\hat{\mathbcal{Y}}_{\rm itp} = \mat[0]\in\R^{\lfloor{{N}\over{T}}\rfloor T \times p}$, $\hat{\mathbcal{U}}_{\rm itp} = \mat[0]\in\R^{\lfloor{{N}\over{T}}\rfloor T \times m}$\;
    $\hat{\mat[Y]}_{\rm fft} \leftarrow \left(\lfloor{{N}\over{T}}\rfloor T\right) \times \hat{\mat[Y]}_{\rm fft}$,
    $\hat{\mat[U]}_{\rm fft} \leftarrow \left(\lfloor{{N}\over{T}}\rfloor T\right) \times \hat{\mat[U]}_{\rm fft}$\;
    \For{$k = 1,\ldots,T$}{
        $(\hat{\mathbcal{Y}}_{\rm itp})_{\textrm{the $k\lfloor{{N}\over{T}}\rfloor^\text{th}$ row}} \leftarrow (\hat{\mat[Y]}_{\rm fft})_{\textrm{the $k^\text{th}$ row}}$\;
        $(\hat{\mathbcal{U}}_{\rm itp})_{\textrm{the $k\lfloor{{N}\over{T}}\rfloor^\text{th}$ row}} \leftarrow (\hat{\mat[U]}_{\rm fft})_{\textrm{the $k^\text{th}$ row}}$\;
    }

    \tcc{Calculate the spectral density}
    $\hat{\mathbcal{R}}_{\rm itp} \leftarrow \op[fft]\left( \mat[Y]_N \right) - \hat{\mathbcal{Y}}_{\rm itp}$,
    $\hat{\mathbcal{T}}_{\rm itp} \leftarrow \op[fft]\left( \mat[U]_N \right) - \hat{\mathbcal{U}}_{\rm itp}$\;

    \tcc{Calculate the correlation matrices}
    \For{$i = 1,\ldots,p$, $j = 1,\ldots,p$, $i \leq j$}{
        $\hat{\mathbcal{r}}_{ij} \leftarrow \op[ifft]\left[ (\hat{\mathbcal{R}}_{\rm itp})_{\textrm{the $i^\text{th}$ column}} .\!* \overline{(\hat{\mathbcal{R}}_{\rm itp})_{\textrm{the $j^\text{th}$ column}}} \right]$\;
    }
    \For{$\indl = 0,1,\ldots,M-1$}{
        \For{$i = 1,\ldots,p$, $j = 1,\ldots,p$}{
            $(\hat{\mat[\Xi]}_{rr}^\indl)_{\textrm{the element at $(i,j)$}} \leftarrow \left(\lfloor{{N}\over{T}}\rfloor T - \indl\right)^{-1} \times (\hat{\mathbcal{r}}_{ij})_{\textrm{the $\indl^\text{th}$ element}}$\;
        }
    }
    $\hat{\mat[\Xi]}_{tt} \leftarrow \left(\lfloor{{N}\over{T}}\rfloor T\right)^{-2} \times (\hat{\mathbcal{T}}_{\rm itp})\tr .\!* \overline{\hat{\mathbcal{T}}_{\rm itp}}$\;

    \Return $\hat{\mat[\Xi]}_{rr}^\indl$, $\hat{\mat[\Xi]}_{tt}^\indl$, $\indl = 0,1,\ldots,M-1$.
\end{algorithm}

\subsubsection{Covariance Matrix Estimation} \label{section_noise_covariance}
Equations \eqref{eq_correlation_definition}-\eqref{eq_disturbance_output} in Remark \ref{remark_ise_stationary_process} shows the relationship between covariance matrices and correlation matrices.
As the auto-correlation of $\vec[z]_k$ is unavailable, the first step is to estimate $\mat[\Xi]_{zz}$.
Vectorizing \eqref{eq_correlation_rr} gives:
\begin{equation} \label{eq_als_vectorized_rr}
    \op[vec](\mat[\Xi]_{rr}^\indl) = (\mat[C]\otimes\mat[C]\mat[A]^\indl)\op[vec](\mat[\Xi]_{zz}), \quad \indl = 1,2,\ldots.
\end{equation}

Define the following notations:
\begin{equation*} \begin{aligned}
        \mathbcal{M} & = \mat[C] \otimes \begin{bmatrix*} (\mat[C]\mat[A])\tr & (\mat[C]\mat[A]^2)\tr \cdots & (\mat[C]\mat[A]^{M-1})\tr \end{bmatrix*}\tr,             \\
        \mathbcal{p} & = \begin{bmatrix*} \op[vec](\mat[\Xi]_{rr}^1)\tr & \op[vec](\mat[\Xi]_{rr}^2)\tr & \cdots & \op[vec](\mat[\Xi]_{rr}^{M-1})\tr \end{bmatrix*}\tr.
    \end{aligned} \end{equation*}
$\mat[\Xi]_{zz}$ can thus be estimated with the following linear regression model:
\begin{equation} \label{eq_als_ls_model}
    \mathbcal{M} \op[vec](\mat[\Xi]_{zz}) = \mathbcal{p}.
\end{equation}

Besides using the least-squares method, semidefinite programming (SDP) can be used to ensure the estimated covariance matrices to be semi-positive definite:
\begin{equation} \begin{array}{@{}c@{}l} \label{eq_als_sdp_model}
    \min\limits_{\mat[\Xi]_{zz}} \quad & {\left\Vert \mathbcal{M} \op[vec](\mat[\Xi]_{zz}) - \mathbcal{p} \right\Vert}_2, \\
    \!\!\! s.t. & \mat[\Xi] \succeq \mat[0], \mat[\Xi] - \mat[A] \mat[\Xi] \mat[A]\tr \succeq \mat[0], \mat[\Xi]_{rr}^0 - \mat[C] \mat[\Xi] \mat[C]\tr \succeq \mat[0],
\end{array} \end{equation}
The value of $\mat[\Sigma]_{\omega\omega}$, $\mat[\Sigma]_{\nu\nu}$ and $\mat[\Sigma]_{\tau\tau}$ can then be calculated by definitions \eqref{eq_correlation_definition}-\eqref{eq_correlation_rr}.
In summary, the method is shown as Algorithm \ref{algorithm_cov} when $\mathbb{F}^{\rm i} = \mathbb{F}^{\rm f}$.

\begin{algorithm}[htb]
    \LinesNumbered \SetAlgoLined
    \caption{Estimate (Nominal) Covariance} \label{algorithm_cov}
    \KwData{ $\hat{\mat[\Xi]}_{rr}^\indl$, $\hat{\mat[\Xi]}_{tt}$ by \eqref{eq_als_correlation_tt}, $\indl = 0,\ldots,M-1$. }
    \KwResult{ $\hat{\mat[\Sigma]}_{\omega\omega}$, $\hat{\mat[\Sigma]}_{\nu\nu}$ and $\hat{\mat[\Sigma]}_{\tau\tau}$. }

    Select $\lceil{{n^2}\over{p^2}}\rceil+1 \leq M \ll N$\;
    Solve SDP \eqref{eq_als_sdp_model} to obtain a feasible solution $\hat{\mat[\Xi]}_{zz}$\;
    $\hat{\mat[\Sigma]}_{\omega\omega} \leftarrow \mat[\Xi]_{zz} - \mat[A]\mat[\Xi]_{zz}\mat[A]\tr$,
    $\mat[\Sigma]_{\nu\nu} \leftarrow \mat[\Xi]_{rr}^0 - \mat[C]\mat[\Xi]_{zz}\mat[C]\tr$,
    and $\mat[\Sigma]_{\tau\tau} \leftarrow \mat[\Xi]_{tt}$\;

    \Return $\hat{\mat[\Sigma]}_{\omega\omega}$, $\mat[\Sigma]_{\nu\nu}$ and $\mat[\Sigma]_{\tau\tau}$.
\end{algorithm}

\subsubsection{Persistent-excitation Condition} \label{section_noise_persistent}
As for the correlation estimation, both the time and frequency-domain estimations are unbiased and asymptotically consistent as $N\to\infty$ theoretically.
In practice, the estimation could be more robust if the first several periods of samples are eliminated to relieve the impact of transient.

Covariance estimation needs to pay more attention:
Define the strong detectability condition firstly:

\begin{assumption}[Strong detectability] \label{assume_strong_detectable}
    $\mat[A]$ is non-singular, and $\mat[C]$ is of full column rank.
\end{assumption}

The consistency of covariance needs the regression \eqref{eq_als_ls_model} to be overdetermined.
Literature \cite{ALS_Zagrobelny_1} gives the sufficient and necessary condition of the identifiability of $\mat[\Xi]_{zz}$:

\begin{lemma}[{{\cite[Theorem 1]{ALS_Zagrobelny_1}}}] \label{lemma_covariance_identifiability}
    Given the DT ARMA process \eqref{eq_ise_state}-\eqref{eq_ise_output}, and assume $\mat[A]$ is Schur and non-singular.
    Then the covariance $\mat[\Sigma]_{\omega\omega}$ and $\mat[\Sigma]_{\nu\nu}$ can be estimated consistently by the steady-state auto-correlation $\mat[\Xi]_{rr}^\indl$, $\indl = 0,1,\ldots$ if and only if $(\mat[C], \mat[A])$ satisfies Assumption \ref{assume_strong_detectable}.
\end{lemma}

It can be concluded that the covariance consistency is fulfilled if Assumption \ref{assume_strong_detectable} is satisfied and choose $M \geq \lceil{{n^2}\over{p^2}}\rceil+1$.
It is also reasonable to select $M \ll N$ from the perspective of computational efficiency and robustness.
A practical example is to take $M = N^{{1}\over{3}}$.

The above discussion can be summarized into the following:

\begin{assumption} \label{assume_pe_stochastic}
    $\lceil{{n^2}\over{p^2}}\rceil+1 \leq M \ll N$.
\end{assumption}

\subsection{Algorithms Summary} \label{section_algorithm_conclusion}
This section gives $7$ algorithms that are proposed for different aims.
Figure \ref{fig_diagram_algorithm} summarizes them and gives the general workflows in practice.

\begin{figure}[htb]
    \centering \includegraphics[width=0.85\hsize]{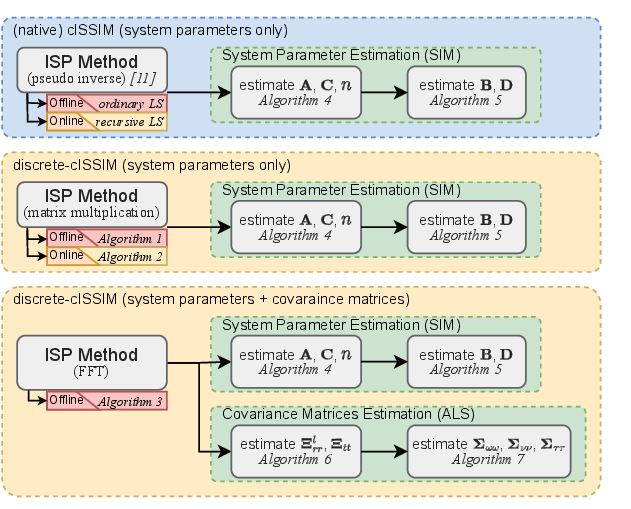}
    \caption{Algorithm Workflow} \label{fig_diagram_algorithm}
\end{figure}

\begin{remark}[computational complexity] \label{remark_complexity}
        The following analysis uses the matrix operation complexity defined in \cite{Numerical_Trefethen_1}, and only the implementations by definition are considered.
    For example, the complexity of the matrix multiplication is regarded as $\mathcal{O}(n^3)$ by definition, not $\mathcal{O}(n^{2.81})$ by Strassen's method \cite{MatrixMultiplication_Strassen_1} or less.
    Assume that there are $N$ samples in the original data, $q$ frequencies used for frequency-based methods, $\bar{n}$ states as the upper bound of system order.
    Other parameters are omitted for brevity.
        Table \ref{table_complexity} shows the computational complexity of discrete-cISSIM and other subspace methods for comparison.

    \begin{table}[htb]
        \centering \caption{Computational Complexity} \label{table_complexity}
        \footnotesize \begin{threeparttable} \begin{tabular}{c|c|l} \hline
            Task & Method & \makecell{Computational \\ Complexity} \\ \hline
            \multirow{5}{*}{\makecell{\scriptsize System Parameters}}
                & cISSIM              & $\mathcal{O}(Nq^2\+q^2\bar{n}\+\bar{n}^3)$ \\ \cline{2-3}
                & \multirow{2}{*}{d-ISIM\tnote{1}} & \multirow{2}{*}{\makecell[l]{$\mathcal{O}(Nq\+q^2\bar{n}\+\bar{n}^3)$ or \\ $\mathcal{O}(N\log T\+q^2\bar{n}\+\bar{n}^3)$}} \\
                & & \\ \cline{2-3}
                & SIM (time)\tnote{2} & $\mathcal{O}(N\bar{n}^2\+\bar{n}^3)$ \\ \cline{2-3}
                & SIM (freq)\tnote{3} & $\mathcal{O}(N\log T\+q^2\bar{n}\+\bar{n}^3)$ \\ \hline
            \multirow{2}{*}{\makecell{\scriptsize System Parameters \& \\ \scriptsize Correlation Matrices}}
                & d-ISIM              & $\mathcal{O}(N\log N\+q^2\bar{n}\+\bar{n}^3)$ \\ \cline{2-3}
                & SIM (time)          & $\mathcal{O}(N\bar{n}^2\+\bar{n}^3)$ \\ \hline
        \end{tabular} \end{threeparttable}
        \begin{tablenotes}
            \footnotesize
            \item{1} refers to the discrete-cISSIM proposed in this paper;
            \item{2} refers to the time-domain SIM \cite{SID_Overschee_2};
            \item{3} refers to the frequency-domain SIM \cite{SysIden_FreqSubspace_McKelvey_1}.
        \end{tablenotes}
    \end{table}

        For the estimation of system parameters only, discrete-cISSIM is the fastest among the $4$ methods.
    Since the covariance matrix estimation has not been discussed before, only the correlation matrices estimation is compared.
    The cISSIM and frequency-domain SIM are invalid in this case, and the discrete-cISSIM have a fair performance compared with time-domain SIM (the relative magnitude of $\log N$ and $\bar{n}^2$).
\end{remark}

\section{Consistency Analysis} \label{section_consistency}
This section introduces the consistency analyses.

\subsection{Consistency of System Parameters Estimation} \label{section_consistency_deterministic}
$\hat{\mat[Y]}^{\rm i}$ and $\hat{\mat[U]}^{\rm i}$ are estimated by simplified least-squares in Section \ref{section_isim_identification}.
The following lemma proves that the estimators are all asymptotically unbiased and consistent as $N\to\infty$:

\begin{lemma} \label{lemma_isim_consistency}
    Suppose Assumptions \ref{assume_stability}-\ref{assume_discrete_periodic} holds, then the estimations generated by each of the Algorithms \ref{algorithm_isim_ols}, \ref{algorithm_isim_rls} and \ref{algorithm_isim_fft} satisfy that $\lim_{N\to\infty}{\hat{\mat[Y]}_N^{\rm i}} = \mat[Y]^{\rm i}$ and $\lim_{N\to\infty}{\hat{\mat[U]}_N^{\rm i}} = \mat[U]^{\rm i}$ with probability $1$.
\end{lemma}
\begin{proof} See \ref{appendix_lemma_isim_consistency}. \end{proof}

The consistency of parameter estimation referred to Section \ref{section_sim_identification} has been proven in cISSIM \cite{cISSIM_CHuang_2}, and Assumption \ref{assume_pe_deterministic} is the derived PE condition.
They are omitted for brevity.
To sum up, the estimations of system parameters are consistent:
\begin{theorem} \label{corollary_determined_problem}
    Suppose Assumptions \ref{assume_stability}-\ref{assume_discrete_periodic} and \ref{assume_pe_deterministic} hold, the Problem \ref{problem_deterministic_identification} is solved by the combination of one type of ISP methods (Algorithm \ref{algorithm_isim_ols}, \ref{algorithm_isim_rls} or \ref{algorithm_isim_fft}) with SIM (Algorithms \ref{algorithm_sim_acn} and \ref{algorithm_sim_bdx}).
\end{theorem}
\begin{proof}
    Given Lemma \ref{lemma_isim_consistency}, the consistency of $\hat{n}$ and the system parameters $(\hat{\mat[A]}, \hat{\mat[B]}, \hat{\mat[C]}, \hat{\mat[D]})$ can be concluded using the same proof as that of Theorem 1 in \cite{cISSIM_CHuang_2}.
\end{proof}

\subsection{Consistency of covariance matrix estimation}
Under the stationary noise assumption, the estimation of auto-correlations are convergent in mean square when $N\to\infty$:

\begin{lemma} \label{lemma_correlation_consistency}
    Suppose Assumptions \ref{assume_stability} and \ref{assume_full_identification} hold, and the $\hat{\mat[Y]}^{\rm i}$ and $\hat{\mat[U]}^{\rm i}$ have been given.
    Then the estimated auto-correlation calculated in time domain (equations \eqref{eq_als_correlation_rr}-\eqref{eq_als_correlation_tt}) or frequency domain (Algorithm \ref{algorithm_correlation}) converges to the true value in mean square as $N\to\infty$:
    \begin{equation*} \begin{aligned}
        & \lim_{N\to\infty}{\E\left[\op[vec]\left( \hat{\mat[\Xi]}_{rr}^\indl \right) \op[vec]\left( \hat{\mat[\Xi]}_{rr}^\indl \right)\tr\right]} = 0, \indl = 0,1,\ldots,M-1, \\
        & \lim_{N\to\infty}{\E\left[\op[vec]\left( \hat{\mat[\Xi]}_{tt} \right) \op[vec]\left( \hat{\mat[\Xi]}_{tt} \right)\tr\right]} = 0.
    \end{aligned} \end{equation*}
\end{lemma}
\begin{proof} See \ref{appendix_lemma_correlation_consistency}. \end{proof}

Finally, it can be concluded that:
\begin{theorem} \label{corollary_covariance_problem}
    Suppose Assumptions \ref{assume_stability}-\ref{assume_full_identification} hold, the Problem \ref{problem_stochastic_identification} is solved by Algorithms \ref{algorithm_isim_fft}, \ref{algorithm_sim_acn}, \ref{algorithm_sim_bdx}, \ref{algorithm_correlation} and \ref{algorithm_cov}.
\end{theorem}
\begin{proof}
        It can be derived from Lemma \ref{lemma_correlation_consistency} that $\mat[\Sigma]_{\tau\tau}$ can be estimated consistently in mean square using Algorithm \ref{algorithm_correlation}.
    Then $\hat{\mathbcal{p}}$ in \eqref{eq_als_ls_model} also converges in mean square.

        As for the estimation of $\mat[\Sigma]_{\omega\omega}$ and $\mat[\Sigma]_{\nu\nu}$, since the SDP \eqref{eq_als_sdp_model} is convex \cite{ConvexOptimization_Boyd_1} and overdetermined from Lemma \ref{lemma_covariance_identifiability}, and for every estimated $\hat{\mathbcal{p}}$, there exists one solution since the feasible region from constraints is always non-empty.
        When $N\to\infty$, the LHS and RHS of equation \eqref{eq_als_ls_model} converges proven by Theorem \ref{corollary_determined_problem}, the solution of SDP $\mat[\Xi]_{zz}$ also converges to the true value, and so do the covariance matrices.
\end{proof}

\begin{remark} \label{remark_nominal_covairance}
    Assumption \ref{assume_strong_detectable} is strong.
    When it is not satisfied, model \eqref{eq_als_ls_model} is underdetermined.
    But it is easy to verify that Algorithm \ref{algorithm_cov} can still generate a \emph{nominal} solution $(\hat{\mat[\Sigma]}_{\omega\omega},\hat{\mat[\Sigma]}_{\nu\nu},\hat{\mat[\Sigma]}_{\tau\tau})$ such that $\hat{\mat[\Xi]}_{rr}^{\rm i}$ calculated by equations \eqref{eq_correlation_zz}-\eqref{eq_correlation_rr}, and $\hat{\mat[\Sigma]}_{\tau\tau}$ still converge to the real value.
\end{remark}

\section{Simulation} \label{section_simulation}

\subsection{Implementation \& Metrics}
The algorithm and experiment are implemented in MATLAB, and the code is available on \href{https://github.com/wyqy/dcissim}{GitHub}.
All the experiments discard the first period of samples to suppress the impact of transient.
The SDP \eqref{eq_als_sdp_model} is solved via CVX toolbox \cite{CVX_CVX_1, CVX_Grant_1}.
Monte-Carlo method are used to evaluate the algorithms: unless otherwise specified, the experiments are conducted multiple times with the same algorithm and excitation but random SUTs and noises.

The experiments are analyzed in the aspects of estimation accuracy and computational speed.
The estimation accuracy is measured by the relative error of norms between the estimations and the true values.
The computational speed is compared by calculating the mean of execution times of multiple experiments.

As for system parameter estimation, denote $G$ as the transfer function form of systems under test (SUTs).
The relative errors of $H_2$ and $H_\infty$ norms are used as the criteria.
The covariance estimation uses two types of relative error criteria: one is of covariance matrices, and the other is of correlation matrices.
The criterion based on $\mat[\Sigma]_{\omega\omega}$ needs compensation because $\mat[\Sigma]_{\omega\omega}$ is not invariant under similarity transformation, which is conducted in two steps:
1) diagonalize $\hat{\mat[A]}$ by $\hat{\mat[T]}$;
2) normalize $\hat{\mat[C]}_i$ by $\hat{\mat[K]}$ which is diagonal and minimizes the error: $\hat{\mat[K]} = \min\limits_{\mat[K]} {\Vert \mat[K] \hat{\mat[T]} \hat{\mat[\Sigma]}_{\omega\omega} \hat{\mat[T]}\tr \mat[K]\tr - \mat[\Sigma]_{\omega\omega} \Vert}_F^2$.

The criteria used in the following experiments are listed in Table \ref{table_metrics}, where the subscript $i$ denotes the index of experiment.

\begin{table}[htb]
    \centering \caption{Estimation Metrics} \label{table_metrics}
    \footnotesize \begin{threeparttable} \begin{tabular}{c|c|l} \hline
        Object & Type & \makecell{Metric} \\ \hline
        \multirow{2}{*}{$G$}
            & $\Vert\cdot\Vert_2$ & $\delta_2 = \frac1E \sum_{i=1}^{E}{\frac{\Vert \hat{G}_i - G_i \Vert_2}{\Vert G_i \Vert_2}}$ \\ \cline{2-3}
            & $\Vert\cdot\Vert_\infty$ & $\delta_\infty = \frac1E \sum_{i=1}^{E}{\frac{\Vert \hat{G}_i - G_i \Vert_\infty}{\Vert G_i \Vert_\infty}}$ \\ \hline
        $\mat[\Sigma]_{\omega\omega}$ & $\Vert\cdot\Vert_F$ & $\widetilde{\eta}_\omega = \frac1E \sum_{i=1}^{E}{\frac{\Vert \hat{\mat[K]}_i \hat{\mat[T]}_i \hat   {\mat[\Sigma]}_{\omega\omega,i} \hat{\mat[T]}_i\tr \hat{\mat[K]}_i\tr - \mat[\Sigma]_  {\omega\omega} \Vert_F}{\Vert \mat[\Sigma]_{\omega\omega} \Vert_F}}$ \\ \hline
        $\mat[\Sigma]_{\nu\nu}$ & $\Vert\cdot\Vert_F$ & $\eta_\nu = \frac1E \sum_{i=1}^{E}{\frac{\Vert \hat{\mat[\Sigma]}_{\nu\nu,i} - \mat[\Sigma]_{\nu\nu} \Vert_F}{\Vert \mat[\Sigma]_{\nu\nu} \Vert_F}}$ \\ \hline
        $\mat[\Sigma]_{\tau\tau}$ & $\Vert\cdot\Vert_F$ & $\eta_\tau = \frac1E \sum_{i=1}^{E}{\frac{\Vert \hat{\mat[\Sigma]}_{\tau\tau,i} - \mat[\Sigma]_{\tau\tau} \Vert_F}{\Vert \mat[\Sigma]_{\tau\tau} \Vert_F}}$ \\ \hline
        $\mat[\Xi]_{rr}$ & $\Vert\cdot\Vert_F$ & $\eta_r = \frac1E \sum_{i=1}^{E}{\sum_{j=0}^{r}{\frac{\Vert \hat{\mat[\Xi]}_{rr,i}^j - \mat[\Xi]_{rr}^j \Vert_F}{\Vert \mat[\Xi]_{rr}^j \Vert_F}}}$ \\ \hline
        $t$ & $\bar{t}$ & $\bar{t} = \frac1E \sum_{i=1}^{E}{t_i}$ \\ \hline
    \end{tabular} \end{threeparttable}
\end{table}

\subsection{System Parameter Estimation using SISO SUTs}
This part is to evaluate the effect of discrete-cISSIM in system parameter estimation.
The excitation uses periodic \textbf{chirp} signal.
The angular frequency moves from $0$ to $0.4\pi\;rad/s$ linearly.
Set the SNRs as $30\rm dB$.
Set $T = 3000$, $N = 720,000$, $n = 4$, $m = 1$, $p = 1$ (SISO SUTs).

Two types of discrete-cISSIM using $\mathbb{F}^{\rm i} = \mathbb{F}^{\rm f}$ (full) and $\mathbb{F}^{\rm i} \subsetneq \mathbb{F}^{\rm f}$ (reduced) are compared.
The reduced $\mathbb{F}^{\rm i}$ is made by selecting the frequencies uniformly in the band of chirp signal
Other methods for comparison include prediction error minimization method (PEM) \cite{SystemIdentification_Ljung_1}, time-domain SIM \cite{SID_Overschee_2} and frequency-domain SIM \cite{SysIden_FreqSubspace_McKelvey_1}.
They are implemented by MATLAB \cite{SystemIdentificationToolbox_Ljung_1}.
All the hyperparameters such as the system orders in SIMs and discrete-cISSIMs are fixed to the true value to simplify the comparison.
The results based on $100$ random SISO SUTs are shown as Table \ref{table_siso_results}.

\begin{table}[htb]
    \centering \caption{Results of SISO SUTs} \label{table_siso_results}
    \footnotesize \begin{threeparttable} \begin{tabular}{c|c|c|c|c|c} \hline
            \multirow{2}{*}{Method} & \multicolumn{2}{c|}{$\delta_2$} & \multicolumn{2}{c|}{$\delta_\infty$} & $t(s)$ \\ \cline{2-6}
                                & mean   & std    & mean   & std    & mean   \\ \hline
            d-ISIM (f)\tnote{1} & 0.0193 & 0.0007 & 0.0200 & 0.0016 & 0.1806 \\ \hline
            d-ISIM (r)\tnote{2} & 0.0177 & 0.0037 & 0.0204 & 0.0066 & 0.0074 \\ \hline
            tfest               & 0.0146 & 0.0007 & 0.0177 & 0.0015 & 0.2564 \\ \hline
            SIM (time)          & 0.0079 & 0.0013 & 0.0089 & 0.0019 & 0.2744 \\ \hline
            SIM (freq)          & 0.0386 & 0.0179 & 0.0416 & 0.0211 & 0.0397 \\ \hline
        \end{tabular} \end{threeparttable}
    \begin{tablenotes}
        \footnotesize
        \item{1} refers to discrete-cISSIM with full $\mathbb{F}^{\rm i}$;
        \item{2} refers to discrete-cISSIM with reduced $\mathbb{F}^{\rm i}$.
    \end{tablenotes}
\end{table}

The $H_2$ and $H_\infty$ criteria show that discrete-cISSIM is kind of worse than classical time-domain methods, but is better than frequency-domain method in the system parameter estimation accuracy.
As for the computational efficiency, the discrete-cISSIM with reduced $\mathbb{F}^{\rm i}$ is the fastest among all the methods with nearly no performance loss.
Figure \ref{fig_siso_scatter} shows the value distribution of the two metrics for the $100$ experiments.
Figure \ref{fig_siso_recursive} shows the $H_2$ norm versus iteration of a plant using online discrete-cISSIM, which shows that the estimation is convergent into the true value.

\begin{figure}[htb]
    \centering \includegraphics[width=0.7\hsize]{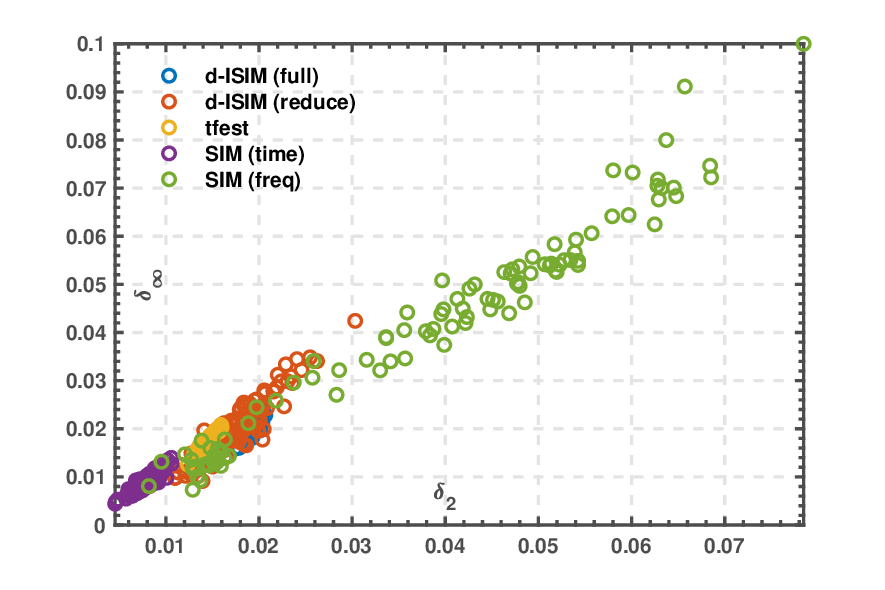}
    \caption{Scatter Plot of $H_2$ \& $H_\infty$ Error} \label{fig_siso_scatter}
\end{figure}

\begin{figure}[htb]
    \centering \includegraphics[width=0.7\hsize]{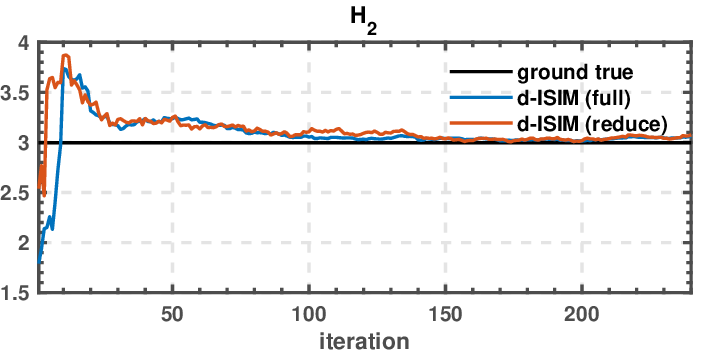}
    \caption{Line Chart of $H_2$ in Online Identification} \label{fig_siso_recursive}
\end{figure}

\subsection{Covariance Matrix Estimation using MIMO SUTs}
This part is to evaluate the effect of discrete-cISSIM in covariance estimation.
One progressive experiment is conducted including multiple iterations, and the latter one has additional $10$ periods of data than the former one.
The excitation uses \textbf{random} signal.
Set the SNRs as $60\rm dB$.
Set $T = 600$, $N_{max} = 240,000$.
The plant is fixed with the following MIMO SUT:
\begin{equation*}
    \mat[A] = \begin{bmatrix} 0.8 & 0 \\ 0 & 0.2 \end{bmatrix},
    \mat[B] = \begin{bmatrix} 1 \\ 1 \end{bmatrix},
    \mat[C] = \begin{bmatrix} 1 & 0 \\ 1 & 1 \end{bmatrix},
    \mat[D] = \begin{bmatrix} 0 & 0 \end{bmatrix},
\end{equation*}
with the covariance \eqref{eq_problem_covariance} which generates $60dB$ noise: $\mat[\Sigma] = \op[diag](2.30, 0.16, 2.30, 3.64, 0.11)\times 10^{-3}$.

Three types of covariance estimation methods are used: the first is the discrete-cISSIM; the second is the ALS with a Kalman observer \cite{ALSVector_Dunik_1}; the third is the subspace identification in time domain which uses innovation model \cite{SID_Overschee_1} to present the covariance.
The innovation model will be transformed into $\hat{\mat[\Xi]}_{rr}^{\rm i}$ here for comparison.
The hyperparameters are also fixed to the true value.
Define $r = 10$ in the correlation criteria of $\mat[\Xi]_{rr}$.
The estimation results are shown as Figure \ref{fig_mimo_covariance} and \ref{fig_mimo_correlation} shows.

\begin{figure}[htb]
    \centering \includegraphics[width=1\hsize]{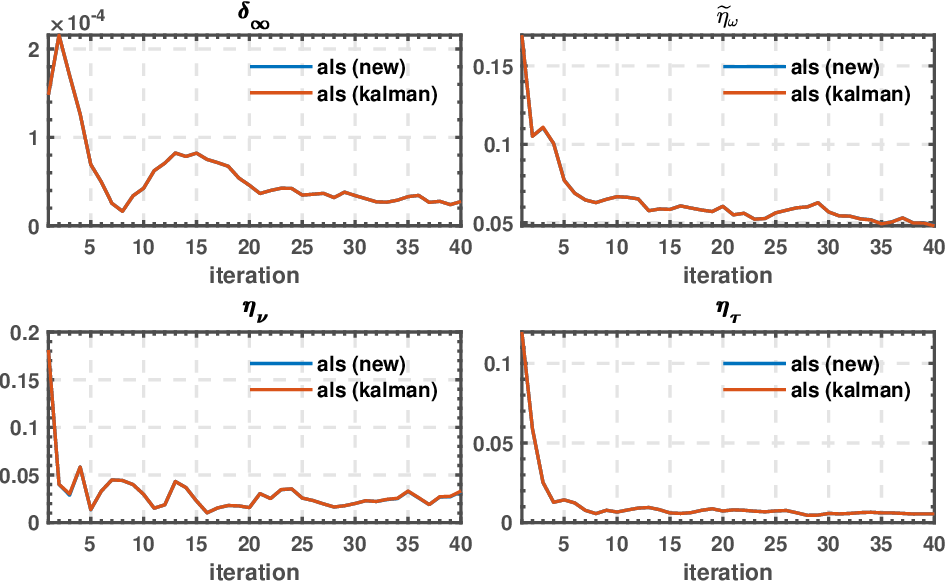}
    \caption{Line Chart of $H_\infty$ \& Covariance Error \tnote{1}\tnote{2}} \label{fig_mimo_covariance}
    \begin{tablenotes}
        \footnotesize
        \item{1} \textit{als (new)} refers to the method proposed in this paper;
        \item{2} \textit{als (kalman)} refers to the method uses Kalman observer.
    \end{tablenotes}
\end{figure}

\begin{figure}[htb]
    \centering \includegraphics[width=0.7\hsize]{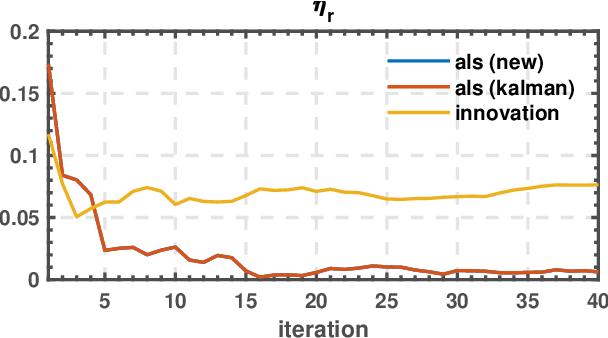}
    \caption{Line Chart of Correlation Error \tnote{1}} \label{fig_mimo_correlation}
    \begin{tablenotes}
        \footnotesize
        \item{1} \textit{Innovation} refers to the innovation model by SIM.
    \end{tablenotes}
\end{figure}

Figure \ref{fig_mimo_covariance} indicates that the Kalman observer is unnecessary when using ISP method.
It also shows that the covariance estimation is asymptotically convergent as $N$ increases.
Figure \ref{fig_mimo_correlation} compare the output correlation with the innovation model identified by subspace identification in time domain.
It shows that the method proposed in this paper can retrieve the noise information better than the traditional black-box method.

\section{Conclusion} \label{section_conclusion}
This paper generalizes the cISSIM from continuous-time to discrete-time, and from multi-sine to arbitrary discrete periodic signal.
The relationship with Fourier analysis is discussed.
A faster implementation is introduced which can decrease the computational burden significantly.
Finally, a novel covariance matrix estimation method in frequency domain is proposed.
The persistent condition and the consistency analysis are given.
The effectiveness of proposed method has also been demonstrated through numerical simulations.
The follow-up research includes improving the covariance estimation, and expanding it into arbitrary excitation.

\appendix

\section{Proof of Lemma \ref{lemma_orthogonal}} \label{appendix_orthogonal_projection}

The solution of excitation system \eqref{eq_excitation_system} can be derived explicitly as follows:
\begin{equation*} \begin{aligned}
        \vec[v]_k^{\rm f} & = \begin{bmatrix*} v_k^0 & v_k^1 & v_k^2 & \cdots & v_k^{2\lfloor{{T}\over{2}}\rfloor} \end{bmatrix*}\tr \in \R^\inds, \\
        v_k^{2r-1}        & = \sin(kr\omega+\psi_r), \quad \widetilde{v}_k^{2r} = \cos(kr\omega+\psi_r),
    \end{aligned} \end{equation*}
where $r = 1, \ldots, \lfloor{{T}\over{2}}\rfloor$.
Rewrite the elements with Euler's formula:
\begin{equation*} \begin{aligned}
        v_k^{2r-1} & = \sin(kr\omega+\psi_r) \\
            & = \frac{\sin\phi_r - \im\cos\phi_r}{2} e^{\im kr\omega} + \frac{\sin\phi_r + \im\cos\phi_r}{2} \overline{e^{\im kr\omega}}; \\
        \widetilde{v}_k^{2r} & = \cos(kr\omega+\psi_r) \\
            & = \frac{\cos\phi_r + \im\sin\phi_r}{2} e^{\im kr\omega} + \frac{\cos\phi_r - \im\sin\phi_r}{2} \overline{e^{\im kr\omega}}.
    \end{aligned} \end{equation*}

Define a function of $k$ as $\mathbcal{e}_\indr(k) = e^{\im k(\indr\mod{T})\omega}$, $\indr \in \Z$, which can construct a functional space $\mathbcal{V}$ endowed with (Hermitian) inner product:
\begin{equation} \label{eq_orthogonalproof_inner_product}
    \left\langle \mathbcal{e}_{\indm}, \mathbcal{e}_{\indn} \right\rangle = \sum_{k=0}^{T-1}{\mathbcal{e}_{\indm}(k)\overline{\mathbcal{e}_{\indn}(k)}}, \quad \indm,\indn\in\Z,
\end{equation}
The following lemma indicates the orthogonality (unitarity) of $\mathbcal{V}$:

\begin{lemma}[{{\cite[Sec.7, Lemma 1.1]{FourierAnalysis_Stein_1}}}] \label{lemma_orthogonal_space}
    \begin{equation*}
        \left\langle \mathbcal{e}_{\indm}, \mathbcal{e}_{\indn} \right\rangle =
        \begin{cases*}
            T, & $\indm = \indn$;    \\
            0, & $\indm \neq \indn$, \\
        \end{cases*}
    \end{equation*}
    where $\indm, \indn \in \Z$.
\end{lemma}

Denote the elements of $(\mat[F]_N^{\rm f})\tr\mat[F]_N^{\rm f} = \left\{f_{\indm,\indn}\right\}_{\indm,\indn = 0}^{2\lfloor{{T}\over{2}}\rfloor}$ by $f_{\indm,\indn} = \sum_{k=0}^{N-1}{v_k^\indm v_k^\indn}$.
When $N = kT$, $k = 1, 2, \ldots$, the following can be derived that for $\indm,\indl=0,\ldots,\lfloor{{T}\over{2}}\rfloor$, when $\indm \neq \indl$,
\begin{equation*}
    f_{\indm,\indn} =
    \begin{cases*}
        kT\cos\!\phi_{\frac{T}{2}}\sin\!\phi_{\frac{T}{2}}, & $\indm, \indl = \{{{T}\over{2}}, {{T}\over{2}}-1\}$, \\
        0, & \text{otherwise when } $\indm \neq \indn$;
    \end{cases*}
\end{equation*}
and when $\indm = \indl$,
\begin{equation*}
    \!\!\!\!f_{\indm,\indn} =
    \begin{cases*}
        \frac{kT}{2}, & $\lfloor{{\indm+1}\over{2}}\rfloor \neq \frac{T}{2}$; \\
        kT\sin^2\!\phi_{\frac{T}{2}}, & $\lfloor{{\indm+1}\over{2}}\rfloor = \frac{T}{2}, \indm \mod{2} = 1$, \\
        kT\cos^2\!\phi_{\frac{T}{2}}, & $\lfloor{{\indm+1}\over{2}}\rfloor = \frac{T}{2}, \indm \mod{2} = 0$,
    \end{cases*}
\end{equation*}
where $f_{0, 0} = \frac{kT}{2}$ holds since $\widetilde{\vec[v]}_{0,0} = 1/\sqrt{2}$.

Thus, the conclusion is derived. \qed

\section{Proof of Lemma \ref{lemma_isim_consistency}} \label{appendix_lemma_isim_consistency}
The proof firstly considers Algorithm \ref{algorithm_isim_ols} using unpartitioned representation: $\hat{\mat[R]}^{\rm i} = {2\over{kT}}(\mat[F]_N^{\rm i})\tr \mat[Z]_N$.

Model \eqref{eq_isim_ols_model} can be vectorized into flatten model:
\begin{equation} \label{eq_consistencyproof_flatten}
    \vec[z]_N = \mat[\Psi]_N^{\rm i} \vec[r]^{\rm i} + \mat[\Psi]_N^{\rm c} \vec[r]^{\rm c} + \vec[\delta]_N + \vec[e]_N,
\end{equation}
where $\vec[z]_N = \op[vec](\mat[Z]_N)$, $\mat[\Psi]_N^{\rm i} = \mat[I]_{m+p}\otimes\mat[F]_N^{\rm i}$, $\vec[r]^{\rm i} = \op[vec](\mat[R]^{\rm i})$, and $\mat[\Psi]_N^{\rm c}$, $\vec[r]^{\rm c}$ are defined similarly.
Split the vectorized disturbance into transient response and noises: $\op[vec](\mat[E]_N) = \vec[\delta]_N + \vec[e]_N = \left\{\delta_{i,k}\right\}_{i=0,1,\ldots,m+p}^{k=0,1,\ldots,N} + \left\{e_{i,k}\right\}_{i=0,1,\ldots,m+p}^{k=0,1,\ldots,N}$.

Model \eqref{eq_consistencyproof_flatten} asymptotically converges into the linear regression model:
Corollary \ref{corollary_orthogonal} states that $\mat[\Psi]_N^{\rm c} \vec[r]^{\rm c} = 0$ as $N=kT$, which is ensured by Algorithms \ref{algorithm_isim_ols}-\ref{algorithm_isim_fft}.
The transient response of $i^\text{th}$ port in $k^\text{th}$ sample is $\vec[\delta]_{i,k} = \vec[C]\tr\mat[A]^k\vec[z]_0$.
Considering that matrices' spectral norm is induced by vectors' Euclidean norm, define $\zeta\widetilde{\vec[z]}_0 = \vec[z]_0$ where ${\left\Vert\widetilde{\vec[z]}_0\right\Vert}_2 = 1$, then ${\left\Vert \mat[C]\mat[A]^k\vec[z]_0 \right\Vert}_2 \leq \zeta{\left\Vert \mat[C]\mat[A]^k \right\Vert}_2 \leq \zeta{\left\Vert\mat[C]\right\Vert}_2 {\left\Vert\mat[A]\right\Vert}_2^k$ for all $k=0,1,\ldots$ \cite{Matrix_Horn_1}.
Since $\mat[A]$ is square and Schur, ${\left\Vert\mat[A]\right\Vert}_2 = \rho(\mat[A]) < 1$.
That means ${\left\Vert\vec[\delta]\right\Vert}_F = \lim_{N\to\infty}{(\sum_{i=0}^{N-1}{{\left\Vert\mat[C]\mat[A]^k\vec[z]_0\right\Vert}_2^2})^{1\over2}} < \infty$.

Thus, model \eqref{eq_consistencyproof_flatten} can be simplified into the following linear-regression model:
\begin{equation} \label{eq_consistencyproof_regression}
    \vec[z]_N = \mat[\Psi]_N^{\rm i} \vec[r]^{\rm i} + \vec[e]_N.
\end{equation}
The result of Algorithm \ref{algorithm_isim_ols} differs from the solution to model \eqref{eq_consistencyproof_regression} by an error which is asymptotically converges to zero.
Considering that the noises $\vec[e]_N$ is a stationary DT ARMA process, Lemma 4 in \cite{cISSIM_CHuang_2} has proven that the solution to model \eqref{eq_consistencyproof_regression} converges to the true value $(w.p.1)$.
Therefore, Algorithms \ref{algorithm_isim_ols} converges with probability $1$.
As for Algorithms \ref{algorithm_isim_rls} and \ref{algorithm_isim_fft}, since they are derived from Algorithm \ref{algorithm_isim_ols} when $N=kT$, the conclusion is also valid. \qed

\section{Proof of Lemma \ref{lemma_correlation_consistency}} \label{appendix_lemma_correlation_consistency}
The proof firstly considers the time-domain method \eqref{eq_als_correlation_rr}-\eqref{eq_als_correlation_tt} to calculate the correlation matrices.

Divide the model \eqref{eq_isim_ols_model} into two parts:
\begin{equation} \label{eq_unbiasedproof_division}
    \begin{bmatrix*} \mat[\mathbcal{Y}]_N & \mat[\mathbcal{U}]_N \end{bmatrix*} = \mat[F]_N^{\rm f} \begin{bmatrix*} (\mat[Y]^{\rm f})\tr & (\mat[U]^{\rm f})\tr \end{bmatrix*} + \begin{bmatrix*} \mat[\mathbcal{R}]_N & \mat[\mathbcal{T}]_N \end{bmatrix*}.
\end{equation}

For the brevity of proof, it is assumed that $N$ samples are available for every $\indl$ in equation \eqref{eq_als_correlation_rr}.
Assumption \ref{assume_full_identification} eliminates the complement frequency term, \ref{appendix_lemma_isim_consistency} proves that transient term converges to zero when $N\to\infty$ and is ignored during covariance estimation.
Thus, equations \eqref{eq_als_correlation_rr}-\eqref{eq_als_correlation_tt} are unbiased.

Divide the multivariate model \eqref{eq_isim_ols_model} by different input/output ports:
$\hat{\mat[Y]}^{\rm i}$ and $\hat{\mat[U]}^{\rm i}$ can be estimated by applying least-squares on \eqref{eq_unbiasedproof_division}, which can be divided by different input ports, as the following shown:
\begin{align}
    \vec[\mathbcal{y}]_\indm & = \mat[F]_N^{\rm i} \mathfrak{y}^{\rm i}_\indm + \vec[\mathbcal{r}]_\indm, \indm = 1,\ldots,p, \label{eq_covarianceproof_y_vec} \\
    \vec[\mathbcal{u}]_\indr & = \mat[F]_N^{\rm i} \mathfrak{u}^{\rm i}_\indr + \vec[\mathbcal{t}]_\indr, \indr = 1,\ldots,m, \label{eq_covarianceproof_u_vec}
\end{align}
where $\vec[\mathbcal{y}]_\indm$, $\mathfrak{y}^{\rm i}_\indm$, and $\vec[\mathbcal{r}]_\indm$ are the $i^\text{th}$ vector of $\mat[\mathbcal{Y}]_N$, $(\mat[Y]^{\rm i})\tr$, and $\mat[\mathbcal{R}]_N$ from equation \eqref{eq_unbiasedproof_division} respectively.
$\vec[\mathbcal{u}]_\indr$, $\mathfrak{u}^{\rm i}_\indr$, and $\vec[\mathbcal{t}]_\indr$ are defined similarly.

Define the residual of least-squares solutions in \eqref{eq_covarianceproof_y_vec} and \eqref{eq_covarianceproof_u_vec} as
$\vec[\mathbcal{e}]_\indm^{\rm r} = \vec[\mathbcal{y}]_\indm - \mat[F]_N^{\rm i} \hat{\mathfrak{y}}_\indm^{\rm i} = \mat[M]_N \vec[\mathbcal{r}]_\indm$ and
$\vec[\mathbcal{e}]_\indr^{\rm t} = \mat[M]_N \vec[\mathbcal{t}]_\indr$,
where $\mat[M]_N = \mat[I]_N - \mat[F]_N^{\rm i}(\mat[F]_N^{\rm i})^\dagger = \mat[I]_N - \mat[F]_N^{\rm i}(\mat[F]_N^{\rm i})\tr$ is idempotent and symmetric.
Their mean squares are defined as
$\delta_{\indm,\indn}^{\rm rr} = {1\over N}(\vec[\mathbcal{e}]_\indm^{\rm r})\tr \vec[\mathbcal{e}]_\indn^{\rm r} = {1\over N}(\vec[\mathbcal{r}]_\indm)\tr \mat[M]_N \vec[\mathbcal{r}]_\indn$ and
$\delta_{\indr,\inds}^{\rm tt} = {1\over N}(\vec[\mathbcal{t}]_\indr)\tr \mat[M]_N \vec[\mathbcal{t}]_\inds$ respectively.

In this case, literature \cite{TimeSeries_Bartlett_1} has proven that the covariance matrices of mean squares defined before are convergence: $\cov(\delta_{\indm,\indn}^{\rm rr}) = O(\frac1N)$ and $\cov(\delta_{\indr,\inds}^{\rm tt}) = O(\frac1N)$ holds for every $\indm,\indn = 1,\ldots,p$ and $\indr,\inds = 1,\ldots,m$.
Assumption \ref{assume_pe_stochastic} ensures that there are infinite samples (denote as $K$) as $N\to\infty$.
Hence, $\cov(\hat{\mat[\Xi]}_{rr}^0) = \left\{\cov(\delta_{rr}^{\indm,\indn})\right\}_{\indm,\indn=1}^p$ and $\cov(\hat{\mat[\Xi]}_{tt}) = \left\{\cov(\delta_{\tau\tau}^{\indr,\inds})\right\}_{\indr,\inds=1}^m$ converge to zero asymptotically.
In other words, $\hat{\mat[\Xi]}_{rr}^0$ and $\cov(\hat{\mat[\Xi]}_{tt})$ are consistent in mean square.
This proof can also be generalized into $\hat{\mat[\Xi]}_{rr}^{\rm i}$ when $i=1,\ldots,M-1$ and is omitted for brevity.

Theorem \ref{theorem_wiener_khinchin} shows that the frequency-domain method, Algorithm \ref{algorithm_correlation}, is equivalent to the above time-domain method when $N\to\infty$.
For finite samples $N < \infty$, it is also valid for sampled version: 

Take the output noise as an example.
Disjoint two ports $\vec[\mathbcal{r}]_\indm$ and $\vec[\mathbcal{r}]_\indn$ from the multivariate noises $\vec[r]_N$.
The $\indl^\text{th}$ correlations by time-domain and frequency-domain method are:
\begin{align}
    \hat{\xi}_{\rm t}^\indl &= \frac{\sum_{k=1}^{N-\indl}{\hat{\mathbcal{r}}_\indm(k\+\indl)\overline{\hat{\mathbcal{r}}_\indn(k)}}}{N-\indl}, & &\textrm{(time-domain)} \label{eq_covarianceproof_timedomain} \\
    \hat{\xi}_{\rm f}^\indl &= \frac{\sum_{k=1}^N{\mathbcal{w}(k) e^{-\im\omega\indl k}}}{N-\indl}, & &\textrm{(freq-domain)} \label{eq_covarianceproof_freqdomain}
\end{align}
where $\mathbcal{w}(k) = \left[ \sum\limits_{\indp=1}^N{\hat{\mathbcal{r}}_\indm(\indp) e^{-\im\omega\indp k}} \right] \left[ \sum\limits_{\indq=1}^N{\hat{\mathbcal{r}}_\indn(\indq) e^{-\im\omega\indq k}} \right]$.

Equation \eqref{eq_covarianceproof_freqdomain} can be simplified into:
\begin{equation*}
    \hat{\xi}_{\rm f}^\indl = \frac{1}{N\-\indl} \sum_{\indp=1}^N{\sum_{\indq=1}^N{ \hat{\mathbcal{r}}_\indm(\indp) \hat{\mathbcal{r}}_\indn(\indq) \sum_{k=1}^N{e^{-\im\omega(\indl-\indp+\indq)k}}}}.
\end{equation*}
Lemma \ref{lemma_orthogonal_space} indicates that the RHS is not null only when $\indl-\indp+\indq = 0$, which is $\indp = \indl+\indq$.
Thus, $\hat{\xi}_{\rm f}^\indl = \hat{\xi}_{\rm t}^\indl$, and the frequency-domain method is equivalent to the time-domain one. \qed

\bibliographystyle{unsrt}
\bibliography{discrete_cissim_derivation_overleaf}

\noindent \textbf{Jingze You}
received the B.S. degree from Tongji University in electronic and information engineering in 2021.
He is currently pursuing the M.E. degree in electrical and information engineering.
His research interests include system identification and state estimation.
\\ \hspace*{\fill} \\\par

\noindent \textbf{Chao Huang}
received the B.S., M.S., and Ph.D. degrees from Zhejiang University, in 2010, 2012 and 2015, respectively, all in Electrical Engineering.
In 2016, he was a post-doctoral research fellow at the School of Engineering, the Australian National University.
From 2017 to 2019, he has been with the School of Automation, Hangzhou Dianzi University, as a lecturer.
He is now with the College of Electronic and Information Engineering, Tongji University, where he is currently an assistant professor.
His research interests include system identification, nonlinear and adaptive control and multi-agent systems.
\\ \hspace*{\fill} \\\par

\noindent \textbf{Hao Zhang}
received the B.Sc. degree in automatic control from the Wuhan University of Technology, Wuhan, China, in 2001, and the Ph.D. degree in control theory and control engineering from the Huazhong University of Science and Technology, Wuhan, in 2007.
She is currently a Professor with the School of Electronics and Information Engineering, Tongji University, Shanghai, China.
Her research interests include network-based control systems, multi-agent systems, and vehicle control systems.\par

\end{document}